\documentclass[journal,twoside,web]{IEEEtran}
\IEEEoverridecommandlockouts % This command is only
\usepackage[keeplastbox]{flushend}
\usepackage{graphicx}
\usepackage{color}
\usepackage{subfigure}
\usepackage{algorithm,algpseudocode}
\usepackage{amsmath, amssymb, amsthm, mathtools, accents}
\usepackage{setspace}
\usepackage[switch,pagewise]{lineno}
\newcommand*\patchAmsMathEnvironmentForLineno[1]{%
\expandafter\let\csname old#1\expandafter\endcsname\csname #1\endcsname
\expandafter\let\csname oldend#1\expandafter\endcsname\csname end#1\endcsname
\renewenvironment{#1}%
{\linenomath\csname old#1\endcsname}%
{\csname oldend#1\endcsname\endlinenomath}}%
\newcommand*\patchBothAmsMathEnvironmentsForLineno[1]{%
\patchAmsMathEnvironmentForLineno{#1}%
\patchAmsMathEnvironmentForLineno{#1*}}%
\AtBeginDocument{%
\patchBothAmsMathEnvironmentsForLineno{equation}%
\patchBothAmsMathEnvironmentsForLineno{align}%
\patchBothAmsMathEnvironmentsForLineno{flalign}%
\patchBothAmsMathEnvironmentsForLineno{alignat}%
\patchBothAmsMathEnvironmentsForLineno{gather}%
\patchBothAmsMathEnvironmentsForLineno{multline}%
}
\usepackage{soul,xcolor}
\usepackage{amsmath,amssymb,amsfonts}
\usepackage{algorithm,algpseudocode}
\usepackage{graphicx}
\usepackage{textcomp}
\usepackage{balance}
\usepackage{bbm}
\usepackage{subfigure}
\usepackage{cite}
\usepackage{tikz}
	\usetikzlibrary{shapes,arrows}
	\tikzstyle{frame} = [draw, -latex]
	\tikzstyle{line} = [draw]
	\tikzstyle{line2} = [draw, dashdotted]
	\tikzstyle{line3} = [draw, dashed]
	\tikzstyle{line3UD} = [draw, dashed]
	\tikzstyle{place} = [circle, draw=black, fill=white, thick, inner sep=2pt, minimum size=1mm]
	\tikzstyle{place2} = [circle, draw=black, fill=black, thick, inner sep=2pt, minimum size=1mm]
	\tikzstyle{placeRed} = [circle, draw=red, fill=red, thick, inner sep=2pt, minimum size=1mm]
	\tikzstyle{vertex} = [circle, draw=black, fill=black, thick, inner sep=2pt, minimum size=1mm]
\usepackage{tkz-euclide}

\usetikzlibrary{shapes,arrows} %% It should be defined first

\tikzstyle{decision} = [diamond, draw, fill=blue!20,
    text width=4.5em, text badly centered, node distance=3cm, inner sep=0pt]
\tikzstyle{block1} = [rectangle, draw, text width=8em, text centered, minimum height=4em]
\tikzstyle{block2} = [rectangle, draw, text width=3em, text centered, minimum height=4em]
\tikzstyle{block3} = [rectangle, draw, text width=11em, text centered, minimum height=12em, dashed,black]
\tikzstyle{block4} = [rectangle, draw, text width=11em, text centered, minimum height=18em, dashed,black]
\tikzstyle{block5} = [rectangle, draw, text width=11em, text centered, minimum height=32em, dashed,black]
\tikzstyle{block6} = [rectangle, draw, text width=11em, text centered, minimum height=18.5em, dashed,black]
\tikzstyle{block7} = [rectangle, draw, text width=11em, text centered, minimum height=11.8em, dashed,black]
\tikzstyle{line01} = [draw, -latex']
\tikzstyle{line02} = [draw, latex'-latex']

\makeatletter
\def\BState{\State\hskip-\ALG@thistlm}
\makeatother
\algdef{SE}[DOWHILE]{Do}{doWhile}{\algorithmicdo}[1]{\algorithmicwhile\ #1}%

\algnewcommand\algorithmicswitch{\textit{switch}}
\algnewcommand\algorithmiccase{\textit{case}}
\algnewcommand\algorithmicassert{\texttt{assert}}
\algnewcommand\Assert[1]{\State \algorithmicassert(#1)}%
\algdef{SE}[SWITCH]{Switch}{EndSwitch}[1]{\algorithmicswitch\ #1\ }{\algorithmicend\ \algorithmicswitch}%
\algdef{SE}[CASE]{Case}{EndCase}[1]{\algorithmiccase\ #1}{\algorithmicend\ \algorithmiccase}%
\algtext*{EndSwitch}%
\algtext*{EndCase}%

\usepackage{amsmath,bm}

\newtheorem{corollary}{\bfseries Corollary}
\newtheorem{definition}{\bfseries Definition }
\newtheorem{lemma}{\bfseries Lemma}

\newtheorem{remark}{\bfseries Remark}
\newtheorem{theorem}{\bfseries Theorem}
\newtheorem{condition}{\bfseries Condition}
\newtheorem{assumption}{\bfseries Assumption}

\title{Strong Sign Controllability of Diffusively-Coupled Networks$^{\ast}$}
\author{Nam-Jin Park$^{1}$, Seong-Ho Kwon$^{1}$, Yoo-Bin Bae$^{1}$, Byeong-Yeon Kim$^{2}$, Kevin L. Moore$^{3}$, and Hyo-Sung Ahn$^{1}$

\thanks{${}^{*}$This paper is written based on our previous works \cite{ahn2019topological}, \cite{ahn2019topological_arXiv}. In this version, we correct an error of the works \cite{ahn2019topological}, \cite{ahn2019topological_arXiv}, and provide more generalized and advanced results in topological controllability.}
\thanks{${}^{1}$School of Mechanical Engineering, Gwangju Institute of Science and Technology (GIST), Gwangju, Korea. E-mails: {\tt\small namjinpark@gist.ac.kr; seongho@gist.ac.kr; bub0418@gm.gist.ac.kr; hyosung@gist.ac.kr; }}
\thanks{${}^{2}$Korea Atomic Energy Research Institute, Daejeon, Korea.  E-mail: {\tt\small byeongyeon@kaeri.re.kr}} 
\thanks{${}^{3}$Department of Electrical Engineering, Colorado School of Mines, Golden, CO, USA. E-mails: {\tt\small kmoore@mines.edu}}
}

\begin{document}

\maketitle
%\thispagestyle{empty}
%\pagestyle{empty}

%%%%%%%%%%%%%%%%%%%%%%%%%%%%%%%%%%%%%%%%%%%%%%%%%%%%%%%%%%%%%%%%%%%
\begin{abstract} 
This paper presents several conditions to determine strong sign controllability for diffusively-coupled undirected networks. 
The strong sign controllability is determined by the sign patterns (positive, negative, zero) of the edges.
We first provide the necessary and sufficient conditions for strong sign controllability of basic components, such as path, cycle, and tree. 
Next, we propose a merging process to extend the basic components to a larger graph based on the conditions of the strong sign controllability.
Furthermore, we develop an algorithm of polynomial complexity to find the minimum number of external input nodes while maintaining the strong sign controllability of a network.
%To model the cross-couplings of  multiple topics, we develop a set of rules for opinion updates of a group of agents. The rules are used to design or assign values to the elements of weighting matrices. The cooperative and anti-cooperative couplings are modeled in both the inverse-proportional and proportional feedbacks. The behaviors of opinion dynamics are analyzed using a null space property of state-dependent matrix-weighted Laplacian matrices and a Lyapunov candidate. Various consensus properties of state-dependent matrix-weighted Laplacian matrices are predicted according to the intra-agent network topology and interdependent topical coupling topologies.
%Since it is difficult to find exact consensus for general cases, we consider complete graphs as a special case. In this case, we can find exact locations of eigenvalues. Through a number of numerical simulations, we would attempt to understand the behaviors of misaligned consensus dynamics.
\end{abstract}

%%%%%%%%%%%%%%%%%%%%%%%%%%%%%%%%%%%%%%%%%%%%%%%%%%%%%%%%%%%%%%%%%%%
\section{Introduction}
The network controllability is a hot topic of research, and there are many works in the literature that investigate it from different points of view.
From the viewpoint of a network structure, most works study to use only some information of the edges in a network. 
In particular, the problems of dealing with network controllability for \textit{structured networks} determined by non-zero/zero patterns of edges is called 
\textit{Structural controllability} \cite{lin1974structural} or \textit{Strong structural controllability} \cite{mayeda1979strong}. 
For the structured networks, the notion of \textit{Structural controllability} considers the controllability for \textit{almost choices of edge weights},
whereas the \textit{Strong structural controllability} considers the controllability for \textit{all choices of edge weights}.

\begin{figure}
\centering
{\includegraphics[width=3.5in,height=2.5in,clip,keepaspectratio]{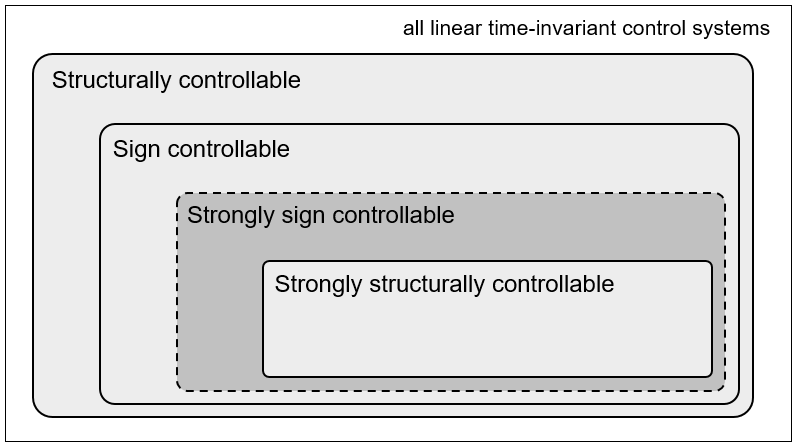}}
\caption{Relationship between controllability concepts for structured and signed networks of linear time-invariant control systems based on \cite{hartung2014sign}.}
\label{relation_controllability}
\end{figure}

The concept of \textit{Sign controllability}, which is the same concept as strong structural controllability of a signed network, was introduced in \cite{johnson1993sign} for the first time,
where the \textit{signed network} is a network corresponding to a graph with positive/negative/zero patterns of edges. 
Many studies have been conducted on the condition of sign controllability using various concepts, such as signed zero forcing set \cite{mousavi2019strong}
and independent strongly connected component (iSCC) \cite{guan2021structural}. 
In particular, in \cite{she2018controllability}, the authors considered the concept of \textit{structural balance} to provide 
the sufficient condition for sign controllability of basic components, e.g., path, tree, and cycle. 
The controllability problem for signed networks considering negative interactions between states can be applied to various fields,
such as social networks \cite{da2018topology}, brain networks \cite{menara2018structural}, and electric circuits \cite{goldstein1979controllability,feng2005structural}.

In the strong structural controllability framework, the problems of \textit{input addition} and \textit{finding minimum inputs} for control efficiency are important issues. 
These problems have been studied in various approaches, such as loopy zero forcing set \cite{jia2020unifying}, structural balance \cite{she2018controllability}, 
zero forcing number \cite{monshizadeh2014zero}, signed zero forcing set \cite{mousavi2019strong}, and constrained matching \cite{chapman2013strong,trefois2015zero}.
In particular, the authors in \cite{mousavi2017structural,trefois2015zero} proved that the problem of finding minimum inputs is NP-hard, thus, minimizing the computational complexity is important issue. 
Recently, the algorithms for finding minimum inputs with a polynomial complexity were proposed in \cite{guan2021structural,chapman2013strong}.

\subsection{Research Flow}
Based on the \textit{Theorem 3.2} in \cite{tsatsomeros1998sign}, we define a new notion of \textit{Strong sign controllability}, which includes the notion of strong structural controllability for a signed network.
Then, we explore the condition of strong sign controllability in diffusively-coupled undirected networks. 
%By strong sign controllability, we mean a sign controllability for all choices of edge weights. 
First, for detailed graph theoretical interpretation of sign controllability, 
we define a concept of \textit{dedicated \& sharing node}, which is a similar concept to the \textit{Dilation} \cite{lin1974structural}, 
where the main difference between \textit{dedicated \& sharing node} and \textit{Dilation} will be stated at a later point.
Then, we interpret the sufficient condition for sign controllability in \textit{Theorem 3.2} of \cite{tsatsomeros1998sign} from the perspective of the dedicated node. 
From \cite{hartung2014sign}, the relationship between the aforementioned concepts of network controllability including
the strong sign controllability can be described as in Fig.~\ref{relation_controllability}.
From \cite{ahn2019topological,ahn2019topological_arXiv}, we extend the results by presenting
the condition for strong sign controllability of a state node unit, and provide the necessary and sufficient condition for strong sign controllability of basic components, such as path, tree, and cycle. 
Similarly, while the authors in \cite{she2018controllability} derived the condition of sign controllability for basic components
based on the concept of \textit{structural balance}, we explore the condition of strong sign controllability from the perspective of \textit{dedicated \& sharing nodes}. 
Then, this paper develops a merging process for sym-pactus type of graphs consisting of basic components, which is a more generalized concept than sym-cactus in \cite{menara2018structural}. 
Furthermore, by interpreting the properties of external input nodes from the perspective of dedicated nodes, 
a state node to have the same properties as the external input node is defined as a \textit{component input node}. 
Based on these two types of input nodes, we propose an algorithm of polynomial complexity to find the minimum number of external input nodes for strong sign controllability.
This allows us to efficiently design strongly sign controllable networks with the minimum number of external input nodes. 

\subsection{Contributions}
Note that this paper is an advanced version of the works on \textit{topological controllability} studied in \cite{ahn2019topological,ahn2019topological_arXiv}.
Although Section~\ref{sec_SSC} in this paper follows the similar logical flow to \cite{ahn2019topological,ahn2019topological_arXiv},
there are some significant difference.
First, there is an error in Corollary 1 of \cite{ahn2019topological,ahn2019topological_arXiv}. 
This error may make the overall results of \cite{ahn2019topological,ahn2019topological_arXiv} misinterpreted, 
thus, this paper corrects this error and makes the results flawless.
Furthermore, we will explore the problem of minimum external input nodes based on the notion of component input nodes,
which is more advanced results compared to \cite{ahn2019topological,ahn2019topological_arXiv}. 
The contributions of this paper are as follows:
\begin{itemize}
\item 
Different from the existing results on the controllability for entire signed networks, e.g., \cite{johnson1993sign,mousavi2019strong,guan2021structural,she2018controllability,tsatsomeros1998sign,she2020energy}, 
for a more detailed analysis from a graph theoretical perspective, 
we investigate the strong sign controllability of \textit{a single state node unit} in a network. 
The aforementioned analysis can determine the controllability of a specific state node in a network, and can be applied to controllable subspace.

\item 
Compared to the existing works on sign controllability using zero forcing set \cite{jia2020unifying,monshizadeh2014zero} and structural balance \cite{she2018controllability}, 
this paper investigates the strong sign controllability based on the \textit{dedicated \& sharing node} as a fundamental concept. 
Note that the concept of \textit{dilation} in \cite{lin1974structural} is for structural controllability, 
while the \textit{sharing node} presented in this paper is interpreted for strong structural controllability.

\item 
This paper interprets the property of an external input node. Based on this observation, 
we present a new notion of \textit{component input node}, which is a state node that has the same property as an external input node in terms of the existence of a \textit{dedicated node}.

\item 
Based on the concept of component input nodes,
this paper presents a merging process for a sym-pactus with a polynomial complexity to satisfy the condition of strong sign controllability with minimum external input nodes.
Note that the sym-pactus defined in this paper is a more generalized concept than sym-cactus introduced in \cite{menara2018structural}.
\end{itemize}

\subsection{Paper organizations}
The paper is organized as follows. In Section~\ref{sec_pre}, the preliminaries and the problems of the sign controllability are formulated. 
In Section~\ref{sec_SSC}, the conditions for the strong sign controllability of basic components and sym-pactus are presented.
In Section~\ref{sec_topo_minimum}, an algorithm for  minimizing the number of external input nodes for the strong sign controllability is developed. 
Topological examples and conclusions are presented in Section~\ref{sec_topo_ex} and Section~\ref{sec_conc}, respectively.

\section{Preliminaries and Problem formulations} \label{sec_pre}
Let us consider undirected networks of diffusively-coupled agents $x_i$ with direct external inputs $u_i$:
\begin{align}
\dot{x}_i = -\sum_{j \in \mathcal{N}_i} a_{ij}(x_i-x_j) + b_i u_i \label{eq_1}
\end{align}
where $ \mathcal{N}_i$ denotes the set of in-neighboring nodes of node $i$, 
$a_{ij}$ are diffusive couplings between node $i$ and $j$ satisfying $a_{ij}=a_{ji}$, and $b_i$ are external input couplings. 
The Laplacian matrix is defined as $L=\mathcal{A}-\mathcal{D}$, 
where $\mathcal{A}\in\mathbb{R}^{n \times n}$ is the adjacency matrix consisting of diffusive couplings $a_{ij}$, and $\mathcal{D}=diag(\mathcal{A}1_n)\in\mathbb{R}^{n \times n}$.
The diffusively-coupled network given by (\ref{eq_1}) can be represented as the Laplacian dynamics:
\begin{align}
\dot{x} = L x + B u \label{eq_dynamics}
\end{align}
where $x=(x_1,\ldots, x_n)^T\in\mathbb{R}^{n \times 1}$, $u=(u_1, \ldots, u_m)^T\in\mathbb{R}^{m \times 1}$, 
$L \in \Bbb{R}^{n \times n}$ is the Laplacian matrix, and $B  \in \Bbb{R}^{n \times m} $ is  the input matrix. 
Let the diffusively-coupled network matrix corresponding to (\ref{eq_dynamics}) be symbolically written as $T=[L, B]\in\mathbb{R}^{n \times (n+m)}$.
From a network point of view,
the Laplacian matrix $L$ includes the interactions of $n$ state nodes and the input matrix $B$ includes input couplings information between $m$ external input nodes and state nodes.
Hence, there are $n+m$ nodes in the network. 
The direction of interactions between state nodes is undirected, while the direction of interactions from the external input nodes to state nodes is directed. 
Also, we assume that each external input node is coupled with only one state node.

\begin{definition}\label{definition_controllability}
(Controllability) A diffusively-coupled network matrix $T= [L, B]$ given by \eqref{eq_dynamics} is said to be controllable 
if there exists an input vector $u(t)$ that satisfies $x(t^*)\to x^*$ at $t=t^*$ for any desired vector $x^*$.
\end{definition}

We define a family set of sign pattern matrices as $Q(T)$, which has the same sign as the network matrix $T$ in an elementwise fashion.
Thus, for all $T' \in Q(T)$, $T'$ has the same sign as $T$ by elementwise.
We also say that the network matrix $T$ is an $\mathcal{L}$-matrix, if the row vectors of $T'$, $\forall T' \in Q(T)$, are linearly independent.
It is clear that $\text{rank}(T) =n$ if and only if the row vectors of matrix $T$ are linearly independent. 
From the viewpoint of control system design, we assume that the input matrix $B$ is fixed since the matrix $B$ can be designed.
Thus, $Q(T)$ is defined as
\begin{align}
Q(T) := [Q({L}), B]
\end{align}
where $Q(L)$ is a set of sign pattern matrices which have the same sign as $L$. 
We say that the system given in \eqref{eq_dynamics} is controllable if and only if its controllability Gramian matrix has full row rank \cite{rugh1996linear}, 
which is an equivalent condition to \textit{Definition~\ref{definition_controllability}}.
Let us assume that a given network matrix $T= [L, B]$ is sign controllable. 
Then, all $T^{'}\in Q(T)$ are controllable and the controllability Gramian matrices of $T^{'}$ have full row rank.
The controllability Gramian matrix of the network matrix $T= [L, B]$ is given by:
\begin{align}
\mathcal{C}_{L} = [B, LB, L^2B, ..., L^{n-1}B]
\end{align}
Then, the sign controllability of a graph $\mathcal{G}(T)$ can be defined:
\begin{definition}\label{definition_TC}\cite{tsatsomeros1998sign}
(Sign controllability) A graph $\mathcal{G}(T)$ given by \eqref{eq_dynamics} is said to be sign controllable (SC)
if all undirected network matrices $T^{'}\in Q(T)$ are controllable.
\end{definition}
This paper is a kind of interpretation from a \textit{graph} point of view of \cite{tsatsomeros1998sign}. 
The network can be re-defined as a \textit{graph}:
\begin{align}
 \mathcal{G}(T) = (\mathcal{V}, \mathcal{E})
\end{align}
where $T= [L, B]$, the vertex set $\mathcal{V}$ is the union of the set of state nodes and the set of input nodes, 
i.e., $\mathcal{V}=\mathcal{V}^{S}\cup\mathcal{V}^{I}$ satisfying $\mathcal{V}^{S}\cap\mathcal{V}^{I}=\emptyset$, 
and the edge set $\mathcal{E}$ is defined by the interactions between nodes $\mathcal{V}$. 
Fig.~\ref{network_graph_concept} depicts a network and a graph.
It is necessary to differentiate \textit{network} and \textit{graph}. The \textit{network} is a connection of physical interactions between nodes, 
while the graph is a representation of the network using the concepts of vertices and edges.  For example, consider a network shown as Fig.~\ref{network_graph_concept}(a).
Let the Laplacian matrix $L$ and the input matrix $B$ corresponding to the network in Fig.~\ref{network_graph_concept}(a) be given as:
\begin{align}
L= \left[\begin{array}{ccccc}
               -3 & 2 & 0 & 1   \\
              2 & -3 & 1 & 0   \\
              0 & 1 & -3 & 2   \\
              1 & 0 & 2 & -3   \\ \end{array} \right],
B= \left[\begin{array}{ccc}
0 & 0 \\
 1 & 0 \\
 0 & 1 \\
 0 & 0 \end{array} \right]
\end{align}
Then, the interaction information of a graph is defined by the matrices $L$ and $B$. 
That is, the network matrix is described as $T=[t_{ij}]=[L, B]$.
We assume that all state nodes in $\mathcal{V}^{S}$ have self-loop, i.e., $t_{ii}\neq 0$, and there is no edge between the external input nodes. 
%For a set $\alpha$, which is a subset of the set of state nodes, i.e., $\alpha \subseteq \mathcal{V}^{S}$, 
%the set of \textcolor{blue}{in-neighboring nodes} of  $\alpha$ is  denoted as $\mathcal{N}(\alpha)= \bigcup \mathcal{N}_i,~\forall i \in \alpha$.
The graph $\mathcal{G}(T)=(\mathcal{V},\mathcal{E})$ consists of a state graph $\mathcal{G}^S$ and an interaction graph $\mathcal{G}^I$ as $\mathcal{G}(T) = \mathcal{G}^S \cup  \mathcal{G}^I$, 
where $\mathcal{G}^S$ is the subgraph induced by $\mathcal{V}^S$, and $\mathcal{G}^I$ is the graph representing the interactions between $\mathcal{V}^S$ and $\mathcal{V}^I$. 
That is, $\mathcal{G}^S = (\mathcal{V}^S, \mathcal{E}^S)$ and $\mathcal{G}^I = (\mathcal{V}^{I}, \mathcal{E}^{I})$, 
the direction of edges in $\mathcal{E}^{S}$ is undirected, while the direction of edges in $\mathcal{E}^{I}$ is directed such that 
$(i,j)\in\mathcal{E}^{I}$ with $i\in\mathcal{V}^{I}$ and $j\in\mathcal{V}^{S}$.
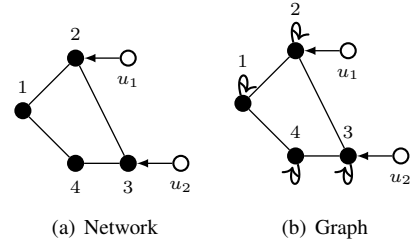
\begin{figure}
\centering
\subfigure[Network]{
\begin{tikzpicture}[scale=0.7]
\node[place, black] (node1) at (-2,0) [label=above:\scriptsize$1$] {};
\node[place, black] (node2) at (-1,1) [label=above:\scriptsize$2$] {};
\node[place, black] (node3) at (0,-1) [label=below:\scriptsize$3$] {};
\node[place, black] (node4) at (-1,-1) [label=below:\scriptsize$4$] {};

\node[place, circle] (node5) at (0,1) [label=below:\scriptsize$u_1$] {};
\node[place, circle] (node6) at (1,-1) [label=below:\scriptsize$u_2$] {};

\draw (node1) [line width=0.5pt] -- node [left] {} (node2);
\draw (node2) [line width=0.5pt] -- node [right] {} (node3);
\draw (node3) [line width=0.5pt] -- node [right] {} (node4);
\draw (node1) [line width=0.5pt] -- node [left] {} (node4);
\draw (node5) [-latex, line width=0.5pt] -- node [right] {} (node2);
\draw (node6) [-latex, line width=0.5pt] -- node [right] {} (node3);
\end{tikzpicture}
}
\subfigure[Graph]{
\begin{tikzpicture}[scale=0.7]

\node[] at (-2,0.8) {\scriptsize$1$};
\node[] at (-1,1.8) {\scriptsize$2$};

\node[place, black] (node1) at (-2,0) [label=below:] {};
\node[place, black] (node2) at (-1,1) [label=below:] {};
\node[place, black] (node3) at (0,-1) [label=above:\scriptsize$3$] {};
\node[place, black] (node4) at (-1,-1) [label=above:\scriptsize$4$] {};

\node[place, circle] (node5) at (0,1) [label=below:\scriptsize$u_1$] {};
\node[place, circle] (node6) at (1,-1) [label=below:\scriptsize$u_2$] {};

\draw (node1) [line width=0.5pt] -- node [left] {} (node2);
\draw (node2) [line width=0.5pt] -- node [right] {} (node3);
\draw (node3) [line width=0.5pt] -- node [right] {} (node4);
\draw (node1) [line width=0.5pt] -- node [left] {} (node4);
\draw (node5) [-latex, line width=0.5pt] -- node [right] {} (node2);
\draw (node6) [-latex, line width=0.5pt] -- node [right] {} (node3);

\draw (node1) edge [loop above,thick] node {$$} (node1);
\draw (node2) edge [loop above,thick] node {$$} (node2);
\draw (node3) edge [loop below,thick] node {$$} (node3);
\draw (node4) edge [loop below,thick] node {$$} (node4);
\end{tikzpicture}
}
\caption{(a) A network with four state nodes and two external input nodes. (b) Graph representation of the network.}
\label{network_graph_concept}
\end{figure}
The following assumptions are necessary for simplicity.
%\begin{assumption} \label{assum_signfixed}
%\st{The values of elements of $L$ can change; but their signs do not change (i.e., sign fixed).}
%\end{assumption}

%\begin{remark}\label{remark_Laplacian}
%\textcolor{blue}{
%The off-diagonal elements of $L$ defined as $a_{ij}$,
%depend on the sign-fixed condition of diagonal elements.   
%Thus, additional edge constraints of off-diagonal elements are required to satisfy the sign-fixed condition of diagonal elements.
%Meanwhile, the network dynamics could be designed as an Adjacency dynamics such that $\dot{x}=Ax+Bu$, where $A=[a_{ij}] \in \Bbb{R}^{n \times n}$. 
%Then, the off-diagonal elements of $A$ have only sign-fixed condition independent of the sign of diagonal elements, i.e., the sign of self-loop. 
%}
%\end{remark}
\begin{assumption} \label{assum_L_matrix}
The network matrix $T$ is an $\mathcal{L}$-matrix.
\end{assumption}
\begin{assumption} \label{assum_accesible}
The graph $\mathcal{G}(T)$ is accessible.
\end{assumption}

%\begin{assumption} \label{assum_selfloop}
%\textcolor{red}{
%Each state node in $\mathcal{V}^{S}$ has a self-loop.}
%\end{assumption}

%Note that since we assume all state nodes in $\mathcal{V}^{S}$ have a self-loop, it is true that $\alpha\subset\mathcal{N}(\alpha)$ for all $\alpha\subseteq\mathcal{V}^{S}$.
The above assumptions are necessary to guarantee the controllability for all  $T' \in Q(T)$.
In \cite{tsatsomeros1998sign}, \textit{Assumption~\ref{assum_accesible}} is required to guarantee \textit{accessibility}\footnote{ In a graph $\mathcal{G}$, accessibility means that  for any $i \in \mathcal{V}^S$, there is a path from $i$ to $j \in \mathcal{V}^I$.} of a graph $\mathcal{G}(T)$. 
If there is no path from an external input node $i\in\mathcal{V}^{I}$ to a state node $j\in\mathcal{V}^{S}$, then the state node $j$ is not controllable.
Based on the above assumptions, the following theorem is a sufficient condition for the sign controllability of a graph.
\begin{theorem} \cite{tsatsomeros1998sign}
Let us suppose that a given network matrix $T=[L,B]$ satisfies 
\textit{Assumption~\ref{assum_L_matrix}} and \textit{Assumption~\ref{assum_accesible}}.
Then, for all $\alpha \subseteq \mathcal{V}^S$ in $\mathcal{G}(T)$, 
if there exists $i\in\mathcal{N}(\alpha)$$\setminus$$\alpha$ such that there exists exactly one edge $(i,j) \in \mathcal{E}$ with $j \in \alpha$,
then the graph $\mathcal{G}(T)$ determined from $T=[L,B]$ is SC. 
\label{theorem_Tsatsomeros}
\end{theorem}

For a graph theoretical interpretation, we classify the node $i\in\mathcal{N}(\alpha)$$\setminus$$\alpha$ as \textit{dedicated nodes} and \textit{sharing nodes}.
For an arbitrary $\alpha$ satisfying $\alpha \subseteq \mathcal{V}^S$,
if a node $i\in\mathcal{N}(\alpha)$$\setminus$$\alpha$ has exactly one edge connected to $\alpha$,
then the node $j$ is called a dedicated node of $\alpha$.
This statement is equivalent to the cardinality condition of $|\mathcal{N}_{i}\cap\alpha|=1$.
On the other hand, a node $i\in\mathcal{N}(\alpha)$$\setminus$$\alpha$ is called a sharing node if the node $j$ has more than one edge connected to $\alpha$,
which is equivalent to $|\mathcal{N}_{i}\cap\alpha|>1$.
Then, a given graph $\mathcal{G}(T)$ is SC if the set $\mathcal{N}(\alpha)$$\setminus$$\alpha$ has at least one dedicated node
for all $\alpha\subseteq\mathcal{V}^{S}$.
Note that the concept of sharing node for sign controllability is similar with the concept of \textit{Dilation} \cite{lin1974structural} for structural controllability.

\begin{definition}\label{definition_candidate_node}
(Dedicated \& Sharing nodes) 
A node $i\in\mathcal{N}(\alpha)$$\setminus$$\alpha$ is a dedicated node of $\alpha$ if it satisfies $|\mathcal{N}_{i}\cap\alpha|\leq1$,
or is a sharing node of $\alpha$ if it satisfies $|\mathcal{N}_{i}\cap\alpha|>1$.
\end{definition}

%\begin{definition}\label{definition_dedicated_node}
%(Dedicated node) A candidate node \textcolor{blue}{$i\in\mathcal{N}(\alpha)$$\setminus$$\alpha$} is a dedicated node 
%if the node $j$ satisfies $|\mathcal{N}_{j}\cap\alpha|=1$.
%\end{definition}

%\begin{definition}\label{definition_sharing_node}
%(Sharing node) A candidate node \textcolor{blue}{$i\in\mathcal{N}(\alpha)$$\setminus$$\alpha$} is a sharing node 
%if the node $j$ satisfies \textcolor{blue}{$|\mathcal{N}_{i}\cap\alpha|>1$}.
%\end{definition}

Note that for connected grpahs, 
a dedicated node $i\in\mathcal{N}(\alpha)$$\setminus$$\alpha$ satisfies $|\mathcal{N}_{i}\cap\alpha|=0$ only when the node $i$ is an external input node.
For example, consider the graph $\mathcal{G}(T)$ depicted in Fig.~\ref{network_graph_concept}(b). 
It is shown that $\mathcal{V}^{S}=\{1,2,3,4\}$ and $\mathcal{V}^{I}=\{u_1,u_2\}$, 
and the edges from $\mathcal{V}^{I}$ to $\mathcal{V}^{S}$ are $(u_1,2), (u_2,3)$.
%According to the condition of $\alpha\subseteq\mathcal{V}^{S}$, the number of possible $\alpha$ is $16$.
Let $\alpha \subseteq \mathcal{V}^S$ be $\alpha=\{1,3\}$.
Then, we obtain $\mathcal{N}(\alpha)$$\setminus$$\alpha=\{2,4,u_2\}$.
Now, we need to check whether the set $\mathcal{N}(\alpha)$$\setminus$$\alpha$ has a dedicated node or not.
For the nodes $2$ and $4$, we obtain $\mathcal{N}_{2}=\{1,2,3,u_1\}$ and $\mathcal{N}_{4}=\{1,3,4\}$, respectively.
Then, it follows that $|\mathcal{N}_{2}\cap\alpha|=|\mathcal{N}_{4}\cap\alpha|=2$, which means that the nodes $2$ and $4$ are the sharing nodes
since those have two edges $(2,1),(2,3)\in\mathcal{E}$ and $(4,1),(4,3)\in\mathcal{E}$ connected to $1,3\in\alpha$, respectively. 
For the external input node $u_2$, we obtain $\mathcal{N}_{u_2}=\emptyset$.
Then, it follows that $|\mathcal{N}_{u_2}\cap\alpha|=0$, which means that the node $u_2$ is a dedicated node of $\alpha$
since the node $u_2$ has exactly one edge $(u_2,3)\in\mathcal{E}$ connected to $3\in\alpha$.
Hence, in case of $\alpha=\{1,3\}$, there exists a dedicated node $u_2$.
In the same way, if the set $\mathcal{N}(\alpha)$$\setminus$$\alpha$ has at least one dedicated node 
for all possible cases of $\alpha$ satisfying $\alpha\subseteq\mathcal{V}^{S}$, then the graph $\mathcal{G}(T)$ is determined as a SC graph.

\begin{remark}\label{remark_meaning_input}
In a graph $\mathcal{G}(T)=(\mathcal{V},\mathcal{E})$, consider a state graph $\mathcal{V}^{S}$ with an external input node $u\in\mathcal{V}^{I}$.
Suppose that the set $\mathcal{N}(\alpha)$$\setminus$$\alpha$ contains the node $u$.
Then, since we assume that an external input node has exactly one out-neighbor state node,
the cardinality condition of the node $u$ always satisfy $|\mathcal{N}_{u}\cap\alpha|=0$,
Hence, there exists at least one dedicated node in $\mathcal{N}(\alpha)$$\setminus$$\alpha$
if $\alpha\subseteq\mathcal{V}^{S}$ contains the node $i\in\mathcal{N}(\mathcal{V}^I)$.
\end{remark}

It follows from \textit{Remark~\ref{remark_meaning_input}} that 
if state nodes in $\mathcal{N}(\mathcal{V}^{I})$ belongs to $\alpha\subseteq\mathcal{V}^{S}$,
at least one dedicated node in $\mathcal{N}(\alpha)$$\setminus$$\alpha$ is guaranteed.
Hence, in \textit{Theorem~\ref{theorem_Tsatsomeros}}, the existence of dedicated nodes for $\alpha \subseteq \mathcal{N}(\mathcal{V}^{I})$ does not need to be considered.
Consequently, the condition of $\alpha\subseteq\mathcal{V}^{S}$ in \textit{Theorem~\ref{theorem_Tsatsomeros}} 
can be reduced to $\alpha \subseteq \mathcal{V}^{S}$$\setminus$$\mathcal{N}(\mathcal{V}^{I})$.

\section{Strongly Sign Controllable Graphs} \label{sec_SSC}
In the previous section, we introduced the sufficient condition for the sign controllability based on \textit{Theorem~\ref{theorem_Tsatsomeros}}.
In this section, we define a \textit{strong sign controllability}, which is a more stronger concept than the sign controllability.
Then, the necessary and sufficient conditions for the strong sign controllability from basic components, e.g., path, tree, and cycle, to a larger graph are provided.
In particular, there exists many works to find the controllability conditions of path, tree, and cycle graphs 
using the concepts of \textit{zero forcing set} \cite{jia2020unifying,mousavi2019strong} and \textit{structural balance} \cite{she2018controllability}.
In this paper, we interpret these existing results on the controllability conditions of basic components presented in \cite{jia2020unifying,mousavi2019strong,she2018controllability}
from the perspective of the strong sign controllability based on the concepts of \textit{dedicated \& sharing node}.

\begin{definition}(Strong sign  controllability)
A graph $\mathcal{G}(T)$ given by \eqref{eq_dynamics} is said to be strongly sign controllable (SSC) if
$\mathcal{G}(T)$ satisfies the condition of \textit{Theorem~\ref{theorem_Tsatsomeros}}.
\end{definition}
Note that if a graph $\mathcal{G}(T)$ is SSC, obviously, the graph $\mathcal{G}(T)$ is SC, while the converse is not satisfied.
Based on \textit{Definition~\ref{definition_candidate_node}} and \textit{Remark~\ref{remark_meaning_input}},
the condition of \textit{Theorem~\ref{theorem_Tsatsomeros}} for the strong sign controllability can be simplified:
\begin{corollary}\label{corollary_theorem1} 
Under the same assumptions as in \textit{Theorem~\ref{theorem_Tsatsomeros}},
the graph $\mathcal{G}(T)$ is SSC
if and only if there exists at least one dedicated node in $\mathcal{N}(\alpha)$$\setminus$$\alpha$
for all $\alpha \subseteq \mathcal{V}^{S}$$\setminus$$\mathcal{N}(\mathcal{V}^{I})$.
\end{corollary}

The above \textit{Corollary~\ref{corollary_theorem1}} provides the necessary and sufficient condition for strong sign controllability from the perspective of dediacted nodes.
For further analysis, we present several definitions inspired by \cite{menara2018structural}:

\begin{definition}\label{definition_sym_path}
(Sym-path) A sym-path is a connected undirected graph with $n\geq1$ nodes, edge set  $\{ (i,j):|i-j|=1\}$, and symmetric weights $a_{ij}=a_{ji}$. 
The Laplacian matrix $L=[l_{ij}]\in\mathbb{R}^{n \times n}$ of a sym-path is defined as:
\begin{align}
     l_{ij} := \left\{
     \begin{matrix}
       a_{ij}\neq0 & \text{if}\,\, |i-j|=1 \\
	  -\sum_{j\in\mathcal{N}_{i}} a_{ij}& \text{if}\,\, i=j\\ 
      a_{ij}=0 & \text{otherwise} \\
     \end{matrix}
     \right.
\end{align}
\end{definition}

\begin{definition}\label{definition_sym_cycle}
(Sym-cycle) A sym-cycle is a connected undirected graph with $n\geq3$ nodes, edge set  $\{ (i,j):|i-j|=1\}\cup\{(1,n),(n,1)\}$, 
and symmetric weights $a_{ij}=a_{ji}$. The Laplacian matrix $L=[l_{ij}]\in\mathbb{R}^{n \times n}$ of a sym-cycle is defined as: 
\begin{align}
     l_{ij} := \left\{
     \begin{matrix}
       a_{ij}\neq0 & \text{if}\,\, |i-j|=1, \,\, \text{or}\\&(i,j)\in\{(1,n),(n,1)\}\\ 
	  -\sum_{j\in\mathcal{N}_{i}} a_{ij}& \text{if}\,\, i=j\\ 
       a_{ij}=0 & \text{otherwise} \\
     \end{matrix}
     \right.
\end{align}
\end{definition}

%\textcolor{red}{
%A path graph is a special case of tree graphs.
%In this paper, among the tree graphs, we say that a graph $\mathcal{G}=(\mathcal{V}, \mathcal{E})$ is a tree graph
%if the node set $\mathcal{V}$ contains at least one node, which has more than two \textit{degree}
%(we denote \textit{degree} of a node $k\in\mathcal{V}$ as $deg(k)$ in $\mathcal{G}$).
%}

\begin{lemma}\label{lemma_path}
Consider a sym-path state graph $\mathcal{G}^{S}=(\mathcal{V}^{S},\mathcal{E}^{S})$
and an interaction graph $\mathcal{G}^{I}=(\mathcal{V}^{I},\mathcal{E}^{I})$ satisfying $|\mathcal{V}^{I}|=1$. 
The graph $\mathcal{G}(T)=\mathcal{G}^{S}\cup\mathcal{G}^{I}$ is SSC
if and only if an external input node is connected to the terminal state node\footnote{In a graph $\mathcal{G}(T)=(\mathcal{V},\mathcal{E})$, we say that a state node in $\mathcal{V}^{S}$ is a terminal state node if its \textit{out-degree} is 1, where the \textit{out-degree} means the number of out-neighbor nodes.}.
%if and only if $\mathcal{G}(T)$ is a path graph.
\end{lemma}

\begin{proof}
Let a state graph $\mathcal{G}^{S}=(\mathcal{V}^{S},\mathcal{E}^{S})$ be a sym-path,
and suppose that there exists a directed edge $(k,l)\in\mathcal{E}$ with $k\in\mathcal{V}^{I}$ and $l\in\mathcal{V}^{S}$. %i.e., $k\in\mathcal{N}(\mathcal{V}^{I})$.
For \textit{if} condition,
let us consider that the state node $l\in\mathcal{V}^{S}$ is a terminal state node. 
In this case, there exists at least one dedicated node $i\in\mathcal{N}(\alpha)$$\setminus$$\alpha$ satisfying $|\mathcal{N}_{i}\cap\alpha|=1$
for all $\alpha\subseteq\mathcal{V}^{S}$$\setminus$$\mathcal{N}(\mathcal{V}^{I})$.
Therefore, the graph $\mathcal{G}(T)$ is SSC.

For \textit{only if} condition,
let us consider that the state node $l\in\mathcal{V}^{S}$ is not a terminal state node. 
In this case, when choosing $\alpha=\mathcal{V}^{S}$$\setminus$$\mathcal{N}(\mathcal{V}^{I})$,
we obtain $\mathcal{N}(\alpha)$$\setminus$$\alpha=\{l\}$.
However, since the node $l$ has two edges $(l,i_1),(l,i_2)\in\mathcal{E}^{S}$ connected to $i_1,i_2\in\alpha$, 
the node $l$ is a sharing node satisfying $|\mathcal{N}_{l}\cap\alpha|=2$.
Therefore, the graph $\mathcal{G}(T)$ is not SSC.
\end{proof}

%path example
For example, Fig.~\ref{network_ex_path}(a) shows a sym-path state graph $\mathcal{G}^{S}$ 
with an external input node $u_1$ connected to a terminal state node $4\in\mathcal{V}^{S}$.
It follows from \textit{Lemma~\ref{lemma_path}} 
that there exists at least one dedicated node $i\in\mathcal{N}(\alpha)$$\setminus$$\alpha$ satisfying $|\mathcal{N}_{i}\cap\alpha|=1$ 
for all $\alpha\subseteq\mathcal{V}^{S}$$\setminus$$\mathcal{N}(\mathcal{V}^{I})$.
Hence, the graph $\mathcal{G}(T)$ depicted in Fig.~\ref{network_ex_path}(a) is SSC.
However, the graph depicted in Fig.~\ref{network_ex_path}(b) shows a sym-path state graph $\mathcal{G}^{S}$ 
with an external input node $u_1$ connected to the state node $3\in\mathcal{V}^{S}$, which is not a terminal state node. 
In this case, when choosing $\alpha=\mathcal{V}^{S}$$\setminus$$\mathcal{N}(\mathcal{V}^{I})=\{ 1,2,4\}$, 
we obtain $\mathcal{N}(\alpha)$$\setminus$$\alpha=\{3\}$.
However, since the node $3$ has two edges $(3,2)$ and $(3,4)$ with $2,4\in\alpha$,
the node $3$ is a sharing node satisfying $|\mathcal{N}_{3}\cap\alpha|=2$.
Hence, the graph $\mathcal{G}(T)$ depicted in Fig.~\ref{network_ex_path}(b) is not SSC.
Note that if $\mathcal{G}^{S}$ is a sym-path, a properly located external input node is sufficient for the graph $\mathcal{G}(T)$ to be SSC, 
i.e., the minimum number of external input node for the strong sign controllability of $\mathcal{G}(T)$ is 1.

\begin{figure}[]
\centering
\subfigure[]{
\begin{tikzpicture}[scale=0.8]
\node[place, black] (node1) at (-2,0) [label=below:\scriptsize$1$] {};
\node[place, black] (node2) at (-1,1) [label=above:\scriptsize$2$] {};
\node[place, black] (node3) at (0,0.5) [label=below:\scriptsize$3$] {};
\node[place, black] (node4) at (1,0) [label=below:\scriptsize$4$] {};
\node[place, circle] (node5) at (2,1) [label=above:\scriptsize$u_1$] {};

\draw (node1) [line width=0.5pt] -- node [left] {} (node2);
\draw (node2) [line width=0.5pt] -- node [right] {} (node3);
\draw (node3) [line width=0.5pt] -- node [left] {} (node4);
\draw (node5) [-latex, line width=0.5pt] -- node [right] {} (node4);
\end{tikzpicture}
}
\subfigure[]{
\begin{tikzpicture}[scale=0.8]
\node[place, black] (node1) at (-2,0) [label=below:\scriptsize$1$] {};
\node[place, black] (node2) at (-1,1) [label=above:\scriptsize$2$] {};
\node[place, black] (node3) at (0,0.5) [label=below:\scriptsize$3$] {};
\node[place, black] (node4) at (1,0) [label=below:\scriptsize$4$] {};
\node[place, circle] (node5) at (1,1.5) [label=above:\scriptsize$u_1$] {};

\draw (node1) [line width=0.5pt] -- node [left] {} (node2);
\draw (node2) [line width=0.5pt] -- node [right] {} (node3);
\draw (node3) [line width=0.5pt] -- node [left] {} (node4);
\draw (node5) [-latex, line width=0.5pt] -- node [right] {} (node3);
\end{tikzpicture}
}
\caption{Sym-path state graphs with one external input node.}
\label{network_ex_path}
%\end{figure}
%\begin{figure}[]
\centering
\subfigure[]{
\begin{tikzpicture}[scale=0.8]
\node[place, black] (node1) at (-2,0) [label=below:\scriptsize$1$] {};
\node[place, black] (node2) at (-1,1) [label=below:\scriptsize$2$] {};
\node[place, black] (node3) at (0,2) [label=below:\scriptsize$3$] {};
\node[place, black] (node4) at (0,0.5) [label=below:\scriptsize$4$] {};
\node[place, black] (node5) at (1,0) [label=below:\scriptsize$5$] {};

\node[place, circle] (node6) at (0.7,2.7) [label=above:\scriptsize$u_1$] {};
\node[place, circle] (node7) at (1.7,0.7) [label=above:\scriptsize$u_2$] {};

\draw (node1) [line width=0.5pt] -- node [left] {} (node2);
\draw (node2) [line width=0.5pt] -- node [right] {} (node3);
\draw (node2) [line width=0.5pt] -- node [left] {} (node5);
\draw (node6) [-latex, line width=0.5pt] -- node [right] {} (node3);
\draw (node7) [-latex, line width=0.5pt] -- node [right] {} (node5);
\end{tikzpicture}
}
\subfigure[]{
\begin{tikzpicture}[scale=0.8]
\node[place, black] (node1) at (-2,0) [label=below:\scriptsize$1$] {};
\node[place, black] (node2) at (-1,1) [label=below:\scriptsize$2$] {};
\node[place, black] (node3) at (0,2) [label=below:\scriptsize$3$] {};
\node[place, black] (node4) at (0,0.5) [label=below:\scriptsize$4$] {};
\node[place, black] (node5) at (1,0) [label=below:\scriptsize$5$] {};

\node[place, circle] (node6) at (0.7,1.2) [label=above:\scriptsize$u_1$] {};
\node[place, circle] (node7) at (1.7,0.7) [label=above:\scriptsize$u_2$] {};

\draw (node1) [line width=0.5pt] -- node [left] {} (node2);
\draw (node2) [line width=0.5pt] -- node [right] {} (node3);
\draw (node2) [line width=0.5pt] -- node [left] {} (node5);
\draw (node6) [-latex, line width=0.5pt] -- node [right] {} (node4);
\draw (node7) [-latex, line width=0.5pt] -- node [right] {} (node5);
\end{tikzpicture}
}

\caption{Tree state graphs with two external input nodes.}
\label{network_ex_tree}
\centering
\subfigure[]{
\begin{tikzpicture}[scale=0.8]
\node[place, black] (node1) at (-2,0) [label=above:\scriptsize$1$] {};
\node[place, black] (node2) at (-1,1) [label=above:\scriptsize$2$] {};
\node[place, black] (node3) at (0,-1) [label=below:\scriptsize$3$] {};
\node[place, black] (node4) at (-1,-1) [label=below:\scriptsize$4$] {};

\node[place, circle] (node5) at (0,1) [label=below:\scriptsize$u_1$] {};
\node[place, circle] (node6) at (1,-1) [label=below:\scriptsize$u_2$] {};

\draw (node1) [line width=0.5pt] -- node [left] {} (node2);
\draw (node2) [line width=0.5pt] -- node [right] {} (node3);
\draw (node3) [line width=0.5pt] -- node [right] {} (node4);
\draw (node1) [line width=0.5pt] -- node [left] {} (node4);
\draw (node5) [-latex, line width=0.5pt] -- node [right] {} (node2);
\draw (node6) [-latex, line width=0.5pt] -- node [right] {} (node3);
\end{tikzpicture}
}
\subfigure[]{
\begin{tikzpicture}[scale=0.8]
\node[place, black] (node1) at (-2,0) [label=below:\scriptsize$1$] {};
\node[place, black] (node2) at (-1,1) [label=above:\scriptsize$2$] {};
\node[place, black] (node3) at (0,-1) [label=below:\scriptsize$3$] {};
\node[place, black] (node4) at (-1,-1) [label=below:\scriptsize$4$] {};

\node[place, circle] (node5) at (-3,0) [label=below:\scriptsize$u_1$] {};
\node[place, circle] (node6) at (1,0) [label=left:\scriptsize$u_2$] {};

\draw (node1) [line width=0.5pt] -- node [left] {} (node2);
\draw (node2) [line width=0.5pt] -- node [right] {} (node3);
\draw (node3) [line width=0.5pt] -- node [right] {} (node4);
\draw (node1) [line width=0.5pt] -- node [left] {} (node4);
\draw (node5) [-latex, line width=0.5pt] -- node [right] {} (node1);
\draw (node6) [-latex, line width=0.5pt] -- node [right] {} (node3);

\end{tikzpicture}
}
\caption{Sym-cycle state graphs with two external input nodes.}
\label{network_ex_cycle}
%\end{figure}
%\begin{figure}[]

\end{figure}

For further analysis of a larger graph, 
we define a bridge graph $\mathcal{G}_{ij}^{S}$, 
which connects two disjoint state graphs $\mathcal{G}_i$ and $\mathcal{G}_j$.

\begin{definition}\label{definition_Bridge_graph} 
(Bridge graph) 
A bridge graph is a state graph defined as $\mathcal{G}_{ij}^{S}=(\mathcal{V}_{ij}^{S},\mathcal{E}_{ij}^{S}$),
which connects two disjoint state graphs $\mathcal{G}^{S}_{i}$ and $\mathcal{G}^{S}_{j}$ satisfying $|i-j|=1$.
If nodes $k\in\mathcal{V}^{S}_{i}$ and $l\in\mathcal{V}^{S}_{j}$ are connected by an edge $(k,l)$, then $(k,l)\in\mathcal{E}_{ij}^{S}$ and $k,l\in\mathcal{V}_{ij}^{S}$.
We assume that a state node $k\in\mathcal{V}^{S}_{i}$ is connected to a state node $l\in\mathcal{V}^{S}_{j}$ by one-to-one (injective).
Hence, $|\mathcal{E}_{ij}^{S}|$ satisfies the following boundary condition.
\begin{align}\label{bridge_edge_boundary}
1\le |\mathcal{E}_{ij}^{S}|\le min(|\mathcal{V}_{i}^{S}|,|\mathcal{V}_{j}^{S}|), |i-j|=1
\end{align}
\end{definition} 

We say that the graph $\mathcal{G}(T)=(\mathcal{V},\mathcal{E})$ is induced by 
$m$-disjoint components $\mathcal{G}_{i}=(\mathcal{V}_{i},\mathcal{E}_{i})$ with $\mathcal{E}_{ij}$
if $\mathcal{V}={\bigcup}^{m}_{i=1}\mathcal{V}_{i}$ and $\mathcal{E}={\bigcup}^{m}_{i=1}(\mathcal{E}_{i}\cup\mathcal{E}_{ij})$ 
satisfying $\mathcal{V}_{i}\cap\mathcal{V}_{j}=\emptyset$ and $\mathcal{E}_{i}\cap\mathcal{E}_{ij}=\emptyset$ for $i,j \in \{1,...,m\}$ and $|i-j|=1$.
Then, we call the graphs $\mathcal{G}_{i},i\in\{1,...,m\}$ disjoint components of $\mathcal{G}(T)$.

\begin{lemma}\label{lemma_tree}
Consider a tree state graph $\mathcal{G}^{S}=(\mathcal{V}^{S},\mathcal{E}^{S})$
and an interaction graph $\mathcal{G}^{I}=(\mathcal{V}^{I},\mathcal{E}^{I})$ satisfying $|\mathcal{V}^{I}|=m\ge2$. 
Then, the graph $\mathcal{G}(T)=\mathcal{G}^{S}\cup\mathcal{G}^{I}$ can be induced by $m$-disjoint components $\mathcal{G}_{i}$ with $\mathcal{E}_{ij}^{S}$ 
satisfying $|\mathcal{E}_{ij}^{S}|=1$ and $|\mathcal{V}^{I}_{i}|=1$ for $i,j \in \{1,...,m\}$ and $|i-j|=1$.
The graph $\mathcal{G}(T)=\mathcal{G}^{S}\cup\mathcal{G}^{I}$ is SSC
if and only if each disjoint component $\mathcal{G}_{i}, i\in\{1,...,m\}$ satisfies \textit{Lemma~\ref{lemma_path}}.
\end{lemma}

\begin{proof}
For \textit{if} condition,
let us assume that each disjoint component $\mathcal{G}_{i},i\in\{1,...,m\}$ satisfies \textit{Lemma~\ref{lemma_path}}.
Then, since each $\mathcal{G}_{i},i\in\{1,...,m\}$ is SSC, 
the set $\mathcal{N}(\alpha_{i})$$\setminus$$\alpha_{i}$ contains at least one dedicated node
for all $\alpha_{i}\subseteq\mathcal{V}_{i}^{S}$$\setminus$$\mathcal{N}(\mathcal{V}_{i}^{I})$.
Now, consider the merged graph $\mathcal{G}(T)$ with the bridge edge $(k,l)\in\mathcal{E}_{ij}^{S}$, where $k\in\mathcal{V}_{i}^{S}$ and $l\in\mathcal{V}_{j}^{S}$ for $i,j \in \{1,...,m\}$ satisfying $|i-j|=1$.
For the merged graph $\mathcal{G}(T)$ to be SSC, we only need to consider the existence of dedicated nodes in $\mathcal{N}(\alpha)$$\setminus$$\alpha$
when $\alpha\subseteq\mathcal{V}^{S}$$\setminus$$\mathcal{N}(\mathcal{V}^{I})$ contains at least one bridge node. %i.e., $k\in\alpha$ or $l\in\alpha$. 
Because after adding the bridge edge, each set of in-neighboring nodes of a node in $\mathcal{V}^{S}$ 
remains unchanged except for $\mathcal{N}_{k}$ and $\mathcal{N}_{l}$ in $\mathcal{G}(T)$.
However, since each disjoint component $\mathcal{G}_{i},i\in\{1,...,m\}$ is SSC, even if the bridge node $k\in\mathcal{V}_{i}^{S}$ or $l\in\mathcal{V}_{j}^{S}$ belongs to $\alpha\subseteq\mathcal{V}^{S}$$\setminus$$\mathcal{N}(\mathcal{V}^{I})$, the set $\mathcal{N}(\alpha)$$\setminus$$\alpha$ still has at least one dedicated node in $\mathcal{V}_{i}$ or in $\mathcal{V}_{j}$.
It follows that the existence of at least one dedicated node in $\mathcal{N}(\alpha)$$\setminus$$\alpha$ is independent of 
the bridge edge $(k,l)\in\mathcal{E}_{ij}^{S}$ satisfying $|\mathcal{E}_{ij}^{S}|=1$ for all $\alpha\subseteq\mathcal{V}^{S}$$\setminus$$\mathcal{N}(\mathcal{V}^{I})$.

For \textit{only if} condition, in the merged graph $\mathcal{G}(T)$,
let us suppose that a disjoint component $\mathcal{G}_{q}=\mathcal{G}_{q}^{S}\cup\mathcal{G}_{q}^{I},q\in\{1,...,m\}$ does not satisfy \textit{Lemma~\ref{lemma_path}}.
Since the graph $\mathcal{G}(T)$ is a tree graph, there always exist a state node $i\in\mathcal{V}^{S}$, which has out-degree 3.
Then, there always exists a case without a dedicated node in $\mathcal{N}(\alpha)$$\setminus$$\alpha$
when $\alpha=\mathcal{V}^{S}$$\setminus$$\mathcal{N}(\mathcal{V}^{I})$$\setminus$$\{i\}$.
\end{proof}

As an example of \textit{Lemma~\ref{lemma_tree}}, 
the graph $\mathcal{G}(T)$ depicted in Fig.~\ref{network_ex_tree}(a) shows a tree state graph $\mathcal{G}^{S}$ with two external input nodes $u_1,u_2\in\mathcal{V}^{I}$. 
In this case, the graph $\mathcal{G}(T)$ can be induced by 
$2$-disjoint path graphs $\mathcal{G}_{1}$ and $\mathcal{G}_{2}$ with a bridge edge $(2,4)\in\mathcal{E}_{12}$,
i.e., $\mathcal{G}_{1}:u_1\rightarrow 3 \leftrightarrow 2 \leftrightarrow 1$ and $\mathcal{G}_{2}:u_2\rightarrow 5 \leftrightarrow 4$, 
where the symbol $\rightarrow$ and $\leftrightarrow$ are used to denote directions of the connection between nodes.
It follows from \textit{Lemma~\ref{lemma_tree}} that the merged graph $\mathcal{G}(T)$ is SSC since each disjoint component $\mathcal{G}_{1}$ and $\mathcal{G}_{2}$ satisfies \textit{Lemma~\ref{lemma_path}}.
However, the graph $\mathcal{G}(T)$ depicted in Fig.~\ref{network_ex_tree}(b) can not be induced by 2-disjoint path graphs. 
In this case, when $\alpha=\{ 1,3\}$, we obtain $\mathcal{N}(\alpha)$$\setminus$$\alpha=\{2\}$.
But the node $2$ is a sharing node satisfying $|\mathcal{N}_{2}\cap\alpha|>1$, which is connected to the nodes $1,3\in\alpha$,
thus, the graph $\mathcal{G}(T)$ in Fig.~\ref{network_ex_tree}(b) is not SSC.

Note that if $\mathcal{G}^{S}$ is a tree graph, which induced by $m$-disjoint components, $m$ properly located external input nodes are sufficient for the graph $\mathcal{G}(T)$ to be SSC, 
i.e., the minimum number of external input nodes for the strong sign controllability of $\mathcal{G}(T)$ is $m$.
With the result of \textit{Lemma~\ref{lemma_tree}}, the following \textit{Corollary~\ref{corollary_1}} can be directly obtained:

\begin{corollary}\label{corollary_1}
Let two components $\mathcal{G}_{i}$ and $\mathcal{G}_{j}$ be SSC, respectively.
If there exists a bridge graph $\mathcal{G}_{ij}^{S}$ satisfying $|\mathcal{E}_{ij}^{S}|=1$,
then the merged graph $\mathcal{G}(T)=\mathcal{G}_{i}\cup\mathcal{G}_{ij}^{S}\cup\mathcal{G}_{j}$ is SSC, regardless of the location of the bridge edge.
\end{corollary}

The above corollary means that the existence and location of a bridge edge connecting two disjoint components,
are independent of the controllability of the merged graph.

\begin{lemma}\label{lemma_cycle}
Consider a sym-cycle state graph $\mathcal{G}^{S}=(\mathcal{V}^{S},\mathcal{E}^{S})$
and an interaction graph $\mathcal{G}^{I}=(\mathcal{V}^{I},\mathcal{E}^{I})$ satisfying $|\mathcal{V}^{I}|=2$. 
The graph $\mathcal{G}(T)=\mathcal{G}^{S}\cup\mathcal{G}^{I}$ is SSC
if and only if there exists an edge $(k,l)\in\mathcal{E}$ with $ k,l \in\mathcal{N}(\mathcal{V}^{I})$.
\end{lemma}

\begin{proof} 
Let us consider that a graph $\mathcal{G}(T)=\mathcal{G}^{S}\cup\mathcal{G}^{I}$ consists of a sym-cycle state graph $\mathcal{G}^{S}=(\mathcal{V}^{S},\mathcal{E}^{S})$ 
and an interaction graph $\mathcal{G}^{I}=(\mathcal{V}^{I},\mathcal{E}^{I})$ satisfying $|\mathcal{V}^{I}|=2$.
Then, the graph $\mathcal{G}(T)$ can be induced by $2$-disjoint components 
$\mathcal{G}_{1}$ and $\mathcal{G}_{2}$ with $\mathcal{E}_{12}^{S}$ satisfying $|\mathcal{E}_{12}^{S}|=2$.
Also, each disjoint component $\mathcal{G}_{1}$ and $\mathcal{G}_{2}$ satisfies \textit{Lemma~\ref{lemma_path}}.
Thus, the sets $\mathcal{N}(\alpha_{1})$$\setminus$$\alpha_{1}$ and $\mathcal{N}(\alpha_{2})$$\setminus$$\alpha_{2}$ have at least one dedicated node
for all $\alpha_{1}\subseteq\mathcal{V}_{1}^{S}$$\setminus$$\mathcal{N}(\mathcal{V}^{I}_{1})$ and 
$\alpha_{2}\subseteq\mathcal{V}_{2}^{S}$$\setminus$$\mathcal{N}(\mathcal{V}^{I}_{2})$, respectively.
For the merged graph $\mathcal{G}(T)$ to be SSC,
we only need to consider the existence of dedicated nodes in $\mathcal{N}(\alpha)$$\setminus$$\alpha$
when $\alpha\subseteq\mathcal{V}^{S}$$\setminus$$\mathcal{N}(\mathcal{V}^{I})$ contains at least one bridge node.
Because each set of in-neighboring nodes of a node in $\mathcal{V}^{S}$ remains unchanged except for the nodes in $\mathcal{V}_{12}^{S}$.
Now, start from $\mathcal{G}_{1}\cup\mathcal{G}_{2}$, we gradually add two bridge edges step-by-step for check the condition of \textit{Corollary~\ref{corollary_theorem1}}.
Let the bridge edges be $\{(k_1,l_1),(k_2,l_2)\}\in\mathcal{E}^{S}_{12}$, 
where $k_1,k_2\in\mathcal{V}_{1}^{S}$ and $l_1,l_2\in\mathcal{V}_{2}^{S}$ satisfying $k_1,l_1\notin\mathcal{N}(\mathcal{V}^{I})$.

For \textit{if} condition,
consider a merged graph $\mathcal{G}_{1}\cup\mathcal{G}_{2}$ with the bridge edge $(k_1,l_1)\in\mathcal{E}^{S}_{12}$.
It follows from \textit{Corollary~\ref{corollary_1}} that if each $\mathcal{G}_{1}$ and $\mathcal{G}_{2}$ is SSC, 
the merged graph $\mathcal{G}_{1}\cup\mathcal{G}_{2}$ with a bridge edge $(k_1,l_1)\in\mathcal{E}_{12}^{S}$ is SSC.
For the other bridge edge $(k_2,l_2)\in\mathcal{E}^{S}_{12}$,
suppose that the bridge nodes $k_2,l_2$ satisfy $k_2,l_2\in\mathcal{N}(\mathcal{V}^{I})$.
Then, each set of in-neighboring nodes of nodes $k_2$ and $l_2$ always includes each other, i.e., $k_2\in\mathcal{N}_{l_2}$ and $l_2\in\mathcal{N}_{k_2}$.
Hence, if $\alpha\subseteq\mathcal{V}^{S}$$\setminus$$\mathcal{N}(\mathcal{V}^{I})$ contains $k_2$ or $l_2$,
the set $\mathcal{N}(\alpha)$$\setminus$$\alpha$ always contains at least one of the nodes $k_2$ and $l_2$ with $k_2,l_2\in\mathcal{N}(\mathcal{V}^{I})$.
It follows from \textit{Remark~\ref{remark_meaning_input}} that if $\alpha$ contains a node in $\mathcal{N}(\mathcal{V}^{I})$, 
there always exists at least one dedicated node in $\mathcal{N}(\alpha)$$\setminus$$\alpha$.
Therefore, the graph $\mathcal{G}(T)$ is SSC.
For \textit{only if} condition,
let us suppose that $(k_2,l_2)\notin\mathcal{E}$ with $k_2,l_2\in\mathcal{N}(\mathcal{V}^{I})$.
In this case, when choosing $\alpha=\mathcal{V}^{S}$$\setminus$$\mathcal{N}(\mathcal{V}^{I})$,
we obtain $\mathcal{N}(\alpha)$$\setminus$$\alpha=\{k_2,l_2\}$.
However, the nodes $k_2$ and $l_2$ are sharing nodes satisfying $|\mathcal{N}_{k_2}\cap\alpha|=|\mathcal{N}_{l_2}\cap\alpha|=2$. 
Hence, according to \textit{Corollary~\ref{corollary_theorem1}}, the graph $\mathcal{G}(T)$ is not SSC.
\end{proof}

%cycle example 
For example, Fig.~\ref{network_ex_cycle}(a) shows a sym-cycle state graph $\mathcal{G}^{S}$ with $\mathcal{N}(\mathcal{V}^{I})=\{ 2,3\}$ and there exists an edge $(2,3)\in\mathcal{E}$. According to \textit{Lemma~\ref{lemma_cycle}}, the graph $\mathcal{G}(T)=\mathcal{G}^{S}\cup\mathcal{G}^{I}$ is SSC.
However, the graph $\mathcal{G}(T)$ depicted in Fig.~\ref{network_ex_cycle}(b) shows $\mathcal{N}(\mathcal{V}^{I})=\{ 1,3\}$ 
and there is no edge between the nodes $1,3\in\mathcal{N}(\mathcal{V}^{I})$, i.e., $(1,3)\notin\mathcal{E}$.
In this case, when choosing $\alpha=\mathcal{V}^{S}$$\setminus$$\mathcal{N}(\mathcal{V}^{I})=\{ 2,4\}$, 
we obtain $\mathcal{N}(\alpha)$$\setminus$$\alpha=\{1,3\}$.
it is clear that nodes $1,3\in\mathcal{N}(\alpha)$$\setminus$$\alpha$, are sharing nodes satisfying $|\mathcal{N}_{1}\cap\alpha|=|\mathcal{N}_{3}\cap\alpha|=2$.
Hence, according to \textit{Corollary~\ref{corollary_theorem1}}, the graph $\mathcal{G}(T)$ in Fig.~\ref{network_ex_cycle}(b) is not SSC.
% 한개로는 TC조건을 만족하지 못하고, 2개이상있어야함.
Note that if $\mathcal{G}^{S}$ is a sym-cycle, two properly located external input nodes are sufficient for the graph $\mathcal{G}(T)$ to be SSC, 
i.e., the minimum number of external input nodes for the strong sign controllability of $\mathcal{G}(T)$ is $2$.
The result of \textit{Lemma~\ref{lemma_cycle}} can be generalized as:

\begin{corollary}\label{corollary_2}
Let two disjoint components $\mathcal{G}_{i}$ and $\mathcal{G}_{j}$ be SSC with sym-path state graphs $\mathcal{G}^{S}_{i}$ and $\mathcal{G}^{S}_{j}$, respectively.
If there exists a bridge edge $(k,l)\in\mathcal{E}_{ij}^{S}$ satisfying $k,l\in\mathcal{N}(\mathcal{V}^{I})$,
the merged graph $\mathcal{G}_{i}\cup\mathcal{G}_{ij}^{S}\cup\mathcal{G}_{j}$ is SSC, regardless of the existence of an additional bridge edge in $\mathcal{E}_{ij}^{S}$.
\end{corollary}

The above \textit{Corollary~\ref{corollary_2}} contains the condition of strong sign controllability 
for a graph $\mathcal{G}(T)=\mathcal{G}^{S}\cup\mathcal{G}^{I}$ when $\mathcal{G}^{S}$ is a sym-cycle.
Thus, \textit{Lemma~\ref{lemma_cycle}} is a special case of \textit{Corollary~\ref{corollary_2}}.
Furthermore, \textit{Corollary~\ref{corollary_2}} can be applied to a union graph of a sym-cycle $\mathcal{G}_{i}^{S}$ and a bridge graph $\mathcal{G}_{ij}^{S}$.
For example, 
consider a SSC graph $\mathcal{G}_{i}=\mathcal{G}_{i}^{S}\cup\mathcal{G}_{i}^{I}$ with a sym-cycle $\mathcal{G}_{i}^{S}$.
Then, if the merged graph $\bar{\mathcal{G}}_{i}=\mathcal{G}_{i}\cup\mathcal{G}_{ij}^{S}$ can be induced by 
$2$-disjoint components satisfying \textit{Corollary~\ref{corollary_2}} with $\mathcal{E}_{ij}^{S}$,
the merged graph $\bar{\mathcal{G}}_{i}$ is SSC.
To expand our theories into a larger graph, we define a sym-pactus, which consists of disjoint components with bridge graphs.

\begin{definition}\label{definition_sym_pactus} 
(Sym-pactus) A sym-pactus is a connected graph 
defined as $\mathcal{G}^{S}={\bigcup}^{m}_{i=1}(\mathcal{G}_{i}^{S}\cup\mathcal{G}_{ij}^{S})$ for $i,j \in \{1,...,m\}$ and $|i-j|=1$.
A sym-pactus satisfies the following properties.
\begin{enumerate} 
\item $\mathcal{G}^{S}$ is induced by $m$-disjoint components $\mathcal{G}_{i}^{S}$ with $\mathcal{E}_{ij}^{S}$
\item each $\mathcal{G}_{i}^{S}, i\in\{1,...,m\}$, is either sym-path or sym-cycle \newline (if $|\mathcal{V}_{i}^{S}|=1$, $\mathcal{G}_{i}^{S}$ contains no edge, that is, $\mathcal{E}_{i}^{S}=\emptyset)$
\item $\mathcal{G}_{i}^{S}$ and $\mathcal{G}_{j}^{S}$ are connected by at least one bridge edge $(k,l)\in\mathcal{E}_{ij}^{S}$,
where $k\in\mathcal{V}^{S}_{i}$, $l\in\mathcal{V}^{S}_{j}$ satisfying $|i-j|=1$
\end{enumerate}
\end{definition}

Note that the sym-pactus is a more generalized concept than the sym-cactus in \cite{menara2018structural}.
It means that the sym-cactus is a special case of the sym-pactus.
For example, the bridge edges between two disjoint conponents in sym-pactus may be several satisfying \eqref{bridge_edge_boundary}, while the sym-cactus has only one.
Based on the aforementioned lemmas, the following theorem can be established.

\begin{figure}[]
\centering
\subfigure[A state graph $\mathcal{G}^{S}$ ]{
\begin{tikzpicture}[scale=0.5]

\node[] at (3.5,4.7) {\scriptsize$\mathcal{G}_{1}^{S}$};
\node[] at (7,6) {\scriptsize$\mathcal{G}_{2}^{S}$};
\node[] at (10.4,3) {\scriptsize$\mathcal{G}_{3}^{S}$};
\node[] at (14,4.5) {\scriptsize$\mathcal{G}_{4}^{S}$};

\node[red] at (6.1,3) {\scriptsize$\mathcal{G}_{12}^{S}$};
\node[red] at (10,5) {\scriptsize$\mathcal{G}_{23}^{S}$};
\node[red] at (11.9,3) {\scriptsize$\mathcal{G}_{34}^{S}$};

\node[place, black] (node1) at (1,5) [label=below:\scriptsize$1$] {};
\node[place, black] (node2) at (3,4) [label=below:\scriptsize$2$] {};
\node[place, red] (node3) at (5,3) [label=below:\scriptsize$3$] {};
\node[place, red] (node4) at (7,4) [label=below:\scriptsize$4$] {};
\node[place, black] (node5) at (5,5) [label=below:\scriptsize$5$] {};
\node[place, black] (node6) at (5,7) [label=above:\scriptsize$6$] {};
\node[place, black] (node7) at (7,8) [label=above:\scriptsize$7$] {};
\node[place, black] (node8) at (9,7) [label=above:\scriptsize$8$] {};
\node[place, red] (node9) at (9,5) [label=below:\scriptsize$9$] {};
\node[place, black] (node10) at (9,3) [label=below:\scriptsize$10$] {};
\node[place, red] (node11) at (11,2) [label=below:\scriptsize$11$] {};
\node[place, red] (node12) at (11,4) [label=above:\scriptsize$12$] {};
\node[place, black] (node13) at (13,5) [label=above:\scriptsize$13$] {};
\node[place, black] (node14) at (15,6) [label=above:\scriptsize$14$] {};
\node[place, black] (node15) at (15,4) [label=below:\scriptsize$15$] {};
\node[place, red] (node16) at (13,3) [label=below:\scriptsize$16$] {};

\draw (node1) [line width=0.5pt] -- node [left] {} (node2);
\draw (node2) [line width=0.5pt] -- node [left] {} (node3);
\draw (node3) [red,dashed,line width=0.5pt] -- node [below] {} (node4);
\draw (node4) [line width=0.5pt] -- node [left] {} (node5);
\draw (node5) [line width=0.5pt] -- node [left] {} (node6);
\draw (node6) [line width=0.5pt] -- node [left] {} (node7);
\draw (node7) [line width=0.5pt] -- node [left] {} (node8);
\draw (node8) [line width=0.5pt] -- node [left] {} (node9);
\draw (node4) [line width=0.5pt] -- node [left] {} (node9);
\draw (node9) [red,dashed,line width=0.5pt] -- node [below] {} (node12);
\draw (node10) [line width=0.5pt] -- node [left] {} (node11);
\draw (node10) [line width=0.5pt] -- node [left] {} (node12);
\draw (node11) [line width=0.5pt] -- node [left] {} (node12);
\draw (node11) [red,dashed,line width=0.5pt] -- node [below] {} (node16);
\draw (node13) [line width=0.5pt] -- node [left] {} (node14);
\draw (node14) [line width=0.5pt] -- node [left] {} (node15);
\draw (node15) [line width=0.5pt] -- node [left] {} (node16);
\draw (node13) [line width=0.5pt] -- node [left] {} (node16);

\end{tikzpicture}
}

\subfigure[(\textit{Theorem~\ref{theorem_sym_pactus_sc_condition}}) A SSC graph $\mathcal{G}(T)=\mathcal{G}^{S}\cup\mathcal{G}^{I}$]{
\begin{tikzpicture}[scale=0.5]

\node[] at (3.5,4.7) {\scriptsize$\mathcal{G}_1$};
\node[] at (7,6) {\scriptsize$\mathcal{G}_2$};
\node[] at (10.4,3) {\scriptsize$\mathcal{G}_3$};
\node[] at (14,4.5) {\scriptsize$\mathcal{G}_4$};

\node[red] at (6.1,3) {\scriptsize$\mathcal{G}_{12}^{S}$};
\node[red] at (10,5) {\scriptsize$\mathcal{G}_{23}^{S}$};
\node[red] at (11.9,3) {\scriptsize$\mathcal{G}_{34}^{S}$};

\node[place, black] (node1) at (1,5) [label=below:\scriptsize$1$] {};
\node[place, black] (node2) at (3,4) [label=below:\scriptsize$2$] {};
\node[place, red] (node3) at (5,3) [label=below:\scriptsize$3$] {};
\node[place, red] (node4) at (7,4) [label=below:\scriptsize$4$] {};
\node[place, black] (node5) at (5,5) [label=below:\scriptsize$5$] {};
\node[place, black] (node6) at (5,7) [label=above:\scriptsize$6$] {};
\node[place, black] (node7) at (7,8) [label=above:\scriptsize$7$] {};
\node[place, black] (node8) at (9,7) [label=above:\scriptsize$8$] {};
\node[place, red] (node9) at (9,5) [label=below:\scriptsize$9$] {};
\node[place, black] (node10) at (9,3) [label=below:\scriptsize$10$] {};
\node[place, red] (node11) at (11,2) [label=below:\scriptsize$11$] {};
\node[place, red] (node12) at (11,4) [label=above:\scriptsize$12$] {};
\node[place, black] (node13) at (13,5) [label=above:\scriptsize$13$] {};
\node[place, black] (node14) at (15,6) [label=above:\scriptsize$14$] {};
\node[place, black] (node15) at (15,4) [label=below:\scriptsize$15$] {};
\node[place, red] (node16) at (13,3) [label=below:\scriptsize$16$] {};

\draw (node1) [line width=0.5pt] -- node [left] {} (node2);
\draw (node2) [line width=0.5pt] -- node [left] {} (node3);
\draw (node3) [red,dashed,line width=0.5pt] -- node [below] {} (node4);
\draw (node4) [line width=0.5pt] -- node [left] {} (node5);
\draw (node5) [line width=0.5pt] -- node [left] {} (node6);
\draw (node6) [line width=0.5pt] -- node [left] {} (node7);
\draw (node7) [line width=0.5pt] -- node [left] {} (node8);
\draw (node8) [line width=0.5pt] -- node [left] {} (node9);
\draw (node4) [line width=0.5pt] -- node [left] {} (node9);
\draw (node9) [red,dashed,line width=0.5pt] -- node [below] {} (node12);
\draw (node10) [line width=0.5pt] -- node [left] {} (node11);
\draw (node10) [line width=0.5pt] -- node [left] {} (node12);
\draw (node11) [line width=0.5pt] -- node [left] {} (node12);
\draw (node11) [red,dashed,line width=0.5pt] -- node [below] {} (node16);
\draw (node13) [line width=0.5pt] -- node [left] {} (node14);
\draw (node14) [line width=0.5pt] -- node [left] {} (node15);
\draw (node15) [line width=0.5pt] -- node [left] {} (node16);
\draw (node13) [line width=0.5pt] -- node [left] {} (node16);

\node[place, circle] (node17) at (-0.5,4) [label=below:\scriptsize$u_{1}$] {}; 
\node[place, circle] (node19) at (3.5,6) [label=above:\scriptsize$u_{2}$] {}; 
\node[place, circle] (node20) at (3.5,8) [label=above:\scriptsize$u_{3}$] {}; 
\node[place, circle] (node21) at (7.5,2) [label=below:\scriptsize$u_{4}$] {}; 
\node[place, circle] (node22) at (9.5,1) [label=below:\scriptsize$u_{5}$] {}; 
\node[place, circle] (node23) at (11.5,6) [label=above:\scriptsize$u_{6}$] {}; 
\node[place, circle] (node24) at (13.5,7) [label=above:\scriptsize$u_{7}$] {}; 
\draw (node17) [-latex, line width=0.5pt] -- node [left] {} (node1);
\draw (node19) [-latex, line width=0.5pt] -- node [right] {} (node5);
\draw (node20) [-latex, line width=0.5pt] -- node [right] {} (node6);
\draw (node21) [-latex, line width=0.5pt] -- node [right] {} (node10);
\draw (node22) [-latex, line width=0.5pt] -- node [right] {} (node11);
\draw (node23) [-latex, line width=0.5pt] -- node [right] {} (node13);
\draw (node24) [-latex, line width=0.5pt] -- node [right] {} (node14);

\end{tikzpicture}
}
\subfigure[(\textit{Theorem~\ref{theorem_sym_pactus_nc_condition}}) A SSC graph $\mathcal{G}(T)=\mathcal{G}^{S}\cup\mathcal{G}^{I}$ with $\min{|\mathcal{V}^{IE}|}=4$]{
\begin{tikzpicture}[scale=0.5]

\node[] at (3.5,4.7) {\scriptsize$\mathcal{G}_1$};
\node[] at (7,6) {\scriptsize$\mathcal{G}_2$};
\node[] at (10.4,3) {\scriptsize$\mathcal{G}_3$};
\node[] at (14,4.5) {\scriptsize$\mathcal{G}_4$};

\node[red] at (6.1,3) {\scriptsize$\mathcal{G}_{12}^{S}$};
\node[red] at (10,5) {\scriptsize$\mathcal{G}_{23}^{S}$};
\node[red] at (11.9,3) {\scriptsize$\mathcal{G}_{34}^{S}$};

\node[place, black] (node1) at (1,5) [label=below:\scriptsize$1$] {};
\node[place, black] (node2) at (3,4) [label=below:\scriptsize$2$] {};
\node[place, red] (node3) at (5,3) [label=below:\scriptsize$3$] {};
\node[place, red] (node4) at (7,4) [label=below:\scriptsize$4$] {};
\node[place, black] (node5) at (5,5) [label=below:\scriptsize$5$] {};
\node[place, black] (node6) at (5,7) [label=above:\scriptsize$6$] {};
\node[place, black] (node7) at (7,8) [label=above:\scriptsize$7$] {};
\node[place, black] (node8) at (9,7) [label=above:\scriptsize$8$] {};
\node[place, red] (node9) at (9,5) [label=below:\scriptsize$9$] {};
\node[place, black] (node10) at (9,3) [label=below:\scriptsize$10$] {};
\node[place, red] (node11) at (11,2) [label=below:\scriptsize$11$] {};
\node[place, red] (node12) at (11,4) [label=above:\scriptsize$12$] {};
\node[place, black] (node13) at (13,5) [label=above:\scriptsize$13$] {};
\node[place, black] (node14) at (15,6) [label=above:\scriptsize$14$] {};
\node[place, black] (node15) at (15,4) [label=below:\scriptsize$15$] {};
\node[place, red] (node16) at (13,3) [label=below:\scriptsize$16$] {};

\draw (node1) [line width=0.5pt] -- node [left] {} (node2);
\draw (node2) [line width=0.5pt] -- node [left] {} (node3);
\draw (node3) [red,dashed,line width=0.5pt] -- node [below] {} (node4);
\draw (node4) [line width=0.5pt] -- node [left] {} (node5);
\draw (node5) [line width=0.5pt] -- node [left] {} (node6);
\draw (node6) [line width=0.5pt] -- node [left] {} (node7);
\draw (node7) [line width=0.5pt] -- node [left] {} (node8);
\draw (node8) [line width=0.5pt] -- node [left] {} (node9);
\draw (node4) [line width=0.5pt] -- node [left] {} (node9);
\draw (node9) [red,dashed,line width=0.5pt] -- node [below] {} (node12);
\draw (node10) [line width=0.5pt] -- node [left] {} (node11);
\draw (node10) [line width=0.5pt] -- node [left] {} (node12);
\draw (node11) [line width=0.5pt] -- node [left] {} (node12);
\draw (node11) [red,dashed,line width=0.5pt] -- node [below] {} (node16);
\draw (node13) [line width=0.5pt] -- node [left] {} (node14);
\draw (node14) [line width=0.5pt] -- node [left] {} (node15);
\draw (node15) [line width=0.5pt] -- node [left] {} (node16);
\draw (node13) [line width=0.5pt] -- node [left] {} (node16);

\node[place, circle] (node17) at (-0.5,4) [label=below:\scriptsize$u_{1}$] {}; 
\node[place, circle] (node19) at (3.5,6) [label=above:\scriptsize$u_{2}$] {}; 
\node[place, circle] (node22) at (7.5,2) [label=below:\scriptsize$u_{3}$] {}; 
\node[place, circle] (node23) at (11.5,6) [label=above:\scriptsize$u_{4}$] {}; 
\draw (node17) [-latex, line width=0.5pt] -- node [left] {} (node1);
\draw (node19) [-latex, line width=0.5pt] -- node [right] {} (node5);
\draw (node22) [-latex, line width=0.5pt] -- node [right] {} (node10);
\draw (node23) [-latex, line width=0.5pt] -- node [right] {} (node13);
\end{tikzpicture}
}
\caption{A sym-pactus $\mathcal{G}^{S}={\bigcup}^{4}_{i=1}(\mathcal{G}_{i}^{S}\cup\mathcal{G}_{ij}^{S})$ satisfying $|\mathcal{E}_{ij}^{S}|=1$ and $|i-j|=1$}
\label{network_ex_th2}
\end{figure}
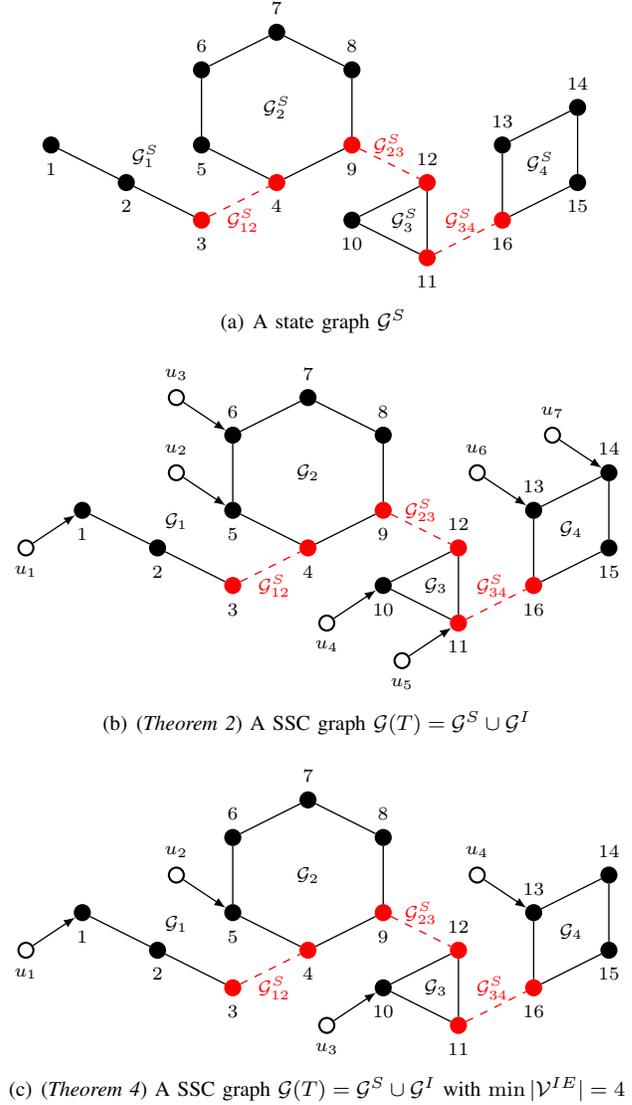

\begin{theorem} \label{theorem_sym_pactus_sc_condition}
Let us suppose that a state graph $\mathcal{G}^{S}$ is a sym-pactus satisfying $|\mathcal{E}_{ij}^{S}|=1$ for $i,j \in \{1,...,m\}$ and $|i-j|=1$.
The graph $\mathcal{G}(T)=\mathcal{G}^{S}\cup\mathcal{G}^{I}$ is SSC if each disjoint component $\mathcal{G}_{i}, i\in \{1,...,m\}$, 
is SSC.
\end{theorem}

\begin{proof}
The \textit{if} condition can be proved by an induction of \textit{Corollary~\ref{corollary_1}}.
Let a state graph $\mathcal{G}^{S}$ be a sym-pactus. 
Then, the sym-pactus $\mathcal{G}^{S}$ can be induced by $m$-disjoint components $\mathcal{G}_{i}^{S}$ with $\mathcal{E}^{S}_{ij}$ for $i,j\in \{1,...,m\}$ and $|i-j|=1$.
Also, each $\mathcal{G}_{i}^{S}$ is either a sym-path or a sym-cycle.
Suppose that each disjoint component $\mathcal{G}_{i}=\mathcal{G}_{i}^{S}\cup\mathcal{G}_{i}^{I}, i\in \{1,...,m\}$ satisfies
\textit{Lemma~\ref{lemma_path}} (sym-path) or \textit{Lemma~\ref{lemma_cycle}} (sym-cycle).
From each disjoint component point of view, it is clear that the set $\mathcal{N}(\alpha_{i})$$\setminus$$\alpha_{i}$ has at least one dedicated node 
for all $\alpha_{i}\subseteq\mathcal{V}_{i}^{S}$$\setminus$$\mathcal{N}(\mathcal{V}_{i}^{I}), i\in \{1,...,m\}$.
Thus, since each SSC component $\mathcal{G}_{i}^{S}$ and $\mathcal{G}_{j}^{S}$ is connected by exactly one bridge edge $(k,l)\in\mathcal{E}_{ij}^{S}$
with $k\in\mathcal{V}^{S}_{i}$ and $l\in\mathcal{V}^{S}_{j}$ for $i,j \in \{1,...,m\}$ and $|i-j|=1$,
by an induction of \textit{Corollary~\ref{corollary_1}}, 
the merged graph $\mathcal{G}(T)$ is SSC, 
i.e., the set $\mathcal{N}(\alpha)$$\setminus$$\alpha$ has at least one dedicated node
for all $\alpha\subseteq {\bigcup}^{m}_{i=1}(\mathcal{V}_{i}^{S}$$\setminus$$\mathcal{N}(\mathcal{V}^{I}_{i}))$, 
which is equivalent to $\alpha\subseteq\mathcal{V}^{S}$$\setminus$$\mathcal{N}(\mathcal{V}^{I})$,
\end{proof}

The above \textit{Theorem~\ref{theorem_sym_pactus_sc_condition}} shows a sufficient condition for the strong sign controllability for sym-pactus, 
which is interpreted from the perspective of each component.
%In other words, it may be possible to satisfy \textit{Theorem~\ref{theorem_Tsatsomeros}} using the fewer number of external input nodes.
As an example of \textit{Theorem~\ref{theorem_sym_pactus_sc_condition}}, 
the state graph $\mathcal{G}^{S}$ depicted in Fig.~\ref{network_ex_th2}(a) is a sym-pactus 
$\mathcal{G}^{S}={\bigcup}^{4}_{i=1}(\mathcal{G}_{i}^{S}\cup\mathcal{G}_{ij}^{S})$ for $i,j\in\{ 1,2,3,4\}$ and $|i-j|=1$.
According to \textit{Lemma~\ref{lemma_path}}, $\mathcal{G}_{1}$ needs an external input node connected to node $1$ to be SSC, 
i.e., $\mathcal{V}_{1}^{I}=\{ u_{1}\}$. 
Since $\mathcal{G}_{2},\mathcal{G}_{3}$, and $\mathcal{G}_{4}$ are sym-cycles, 
each component requires at least two properly located external input nodes to satisfy \textit{Lemma~\ref{lemma_cycle}},
i.e., $\mathcal{V}^{I}_{2}=\{ u_2,u_3\}$, $\mathcal{V}^{I}_{3}=\{ u_4,u_5\}$, $\mathcal{V}^{I}_{4}=\{ u_6,u_7\}$.
These results are shown in Fig.~\ref{network_ex_th2}(b). Hence, the graph $\mathcal{G}(T)=\mathcal{G}^{S}\cup\mathcal{G}^{I}$ requires seven external input nodes 
to satisfy \textit{Theorem~\ref{theorem_sym_pactus_sc_condition}},
i.e., $\mathcal{V}^{I}={\bigcup}^{4}_{i=1}\mathcal{V}_{i}^{I}=\{ u_{1},u_{2},u_{3},u_{4},u_{5},u_{6},u_{7}\}$. 
Note that the locations of the external input nodes are not unique.

\section{Strongly Sign Controllable Graphs with Minimum External Input Nodes} \label{sec_topo_minimum}
In the previous section, we explored the necessary and sufficient conditions for the strong sign controllability of the basic components, 
and these results are extended to a larger graph. 
In this section, we present a condition for the strong sign controllability of a graph from a node point of view.
Then, we present a merging process of polynomial computational complexity to find the minimum number of external input nodes while maintaining the strong sign controllability.
For further analysis of the strong sign controllability from a node point of view,
it is necessary to examine whether a state node is guaranteed at least one dedicated node.
Therefore, we define a SSC state node, which is guaranteed at least one dedicated node in $\mathcal{G}(T)$.

\begin{definition}\label{TC_state_node} 
(SSC state node) %For a sym-pactus $\mathcal{G}^{S}$, 
A set of SSC state nodes in $\mathcal{G}^{S}_{i}$ is symbolically written as $\mathcal{V}_{i}^{SC}$.
A state node $k\in\mathcal{V}_{i}^{S}$ is called a SSC state node
if the set $\mathcal{N}(\alpha)$$\setminus$$\alpha$ has at least one dedicated node 
for all $\alpha\subseteq\mathcal{V}_{i}^{S}$ satisfying $k\in\alpha$.
\end{definition}
Note that it follows from \textit{Remark~\ref{remark_meaning_input}} that the state nodes in $\mathcal{N}(\mathcal{V}^{I})$
are always SSC state nodes.
For convenience, we say that a state node $k\in\mathcal{V}^{S}$ has a dedicated node 
if the set $\mathcal{N}(\alpha)$$\setminus$$\alpha$ has at least one dedicated node 
for all $\alpha\subseteq\mathcal{V}^{S}$ satisfying $k\in\alpha$.
For example, consider the graph depicted in Fig.~\ref{network_ex_cycle}(b).
In this case, the nodes $1,3\in\mathcal{V}^{S}$ are SSC state nodes, 
which are guaranteed a dedicated node from the external input nodes $u_1$ and $u_2$, i.e., $\mathcal{V}^{SC}=\{1,3\}$.
In other word, if $\alpha$ contains at least one SSC state node, the set $\mathcal{N}(\alpha)$$\setminus$$\alpha$ always has at least one dedicated node.
With the above concept of the SSC state node, the following theorem can be established.

\begin{theorem}\label{theorem_TCC}
Consider a sym-pactus $\mathcal{G}^{S}={\bigcup}^{m}_{i=1}\bar{\mathcal{G}}_{i}^{S}$, 
where $\bar{\mathcal{G}}_{i}^{S}=(\bar{\mathcal{V}}_{i}^{S},\bar{\mathcal{E}}_{i}^{S})=\mathcal{G}_{i}^{S}\cup\mathcal{G}_{i(i+1)}^{S}$ for $i\in\{1,...,m\}$.
The graph $\mathcal{G}(T)=\mathcal{G}^{S}\cup\mathcal{G}^{I}$ is SSC 
if and only if the union of SSC state nodes of each component satisfies ${\bigcup}^{m}_{i=1}\bar{\mathcal{V}}_{i}^{SC}=\mathcal{V}^{S}$.
%if and only if the union of SSC state nodes of each component $\bar{\mathcal{G}}_{i}=(\bar{\mathcal{V}}_{i},\bar{\mathcal{E}}_{i})$ 
%satisfies $\bar{\mathcal{V}}_{i}^{SC}=\bar{\mathcal{V}}^{S}_{i}$, i.e., ${\bigcup}^{m}_{i=1}\bar{\mathcal{V}}_{i}^{SC}=\mathcal{V}^{S}$.
\end{theorem}

\begin{proof}
For \textit{if} condition, suppose that the union of SSC state nodes of $\bar{\mathcal{G}}_{i},i\in\{1,...,m\}$ satisfies 
${\bigcup}^{m}_{i=1}\bar{\mathcal{V}}_{i}^{SC}=\mathcal{V}^{S}$.
According to \textit{Definition~\ref{TC_state_node}}, since all state nodes in $\mathcal{V}^{S}$ have at least one dedicated node in $\mathcal{G}(T)$, 
the graph $\mathcal{G}(T)$ is SSC.
For \textit{only if} condition, let us assume that ${\bigcup}^{m}_{i=1}\bar{\mathcal{V}}_{i}^{SC}\subsetneq\mathcal{V}^{S}$.
Then, there exists at least one node $l\in\mathcal{V}^{S}$, which is not a SSC state node, i.e., $l\notin{\bigcup}^{m}_{i=1}\bar{\mathcal{V}}_{i}^{SC}$.
Hence, there exists a case without a dedicated node in $\mathcal{N}(\alpha)$$\setminus$$\alpha$ when $\alpha$ is chosen as $l\in\alpha\subseteq\mathcal{V}^{S}$.
\end{proof}
The above \textit{Theorem~\ref{theorem_TCC}} is an interpretation of the strong sign controllability 
from each component point of view.
With the above observation, the following corollary is directly obtained:

\begin{corollary}\label{corollary_3}
The graph $\mathcal{G}(T)=\mathcal{G}^{S}\cup\mathcal{G}^{I}$ is SSC
if and only if the set of state nodes $\mathcal{V}^{S}$ satisfies $\mathcal{V}^{S}$$=$$\mathcal{V}^{SC}$.
\end{corollary}
For the problem of finding the minimum number of external input nodes, we have to consider the meaning of input nodes carefully.
As shown in \textit{Remark~\ref{remark_meaning_input}}, the property of an external input node is guaranteeing a dedicated node to a state node connected with it.
From the perspective of the dedicated node, similar to the property of external input nodes, if certain structural condition is satisfied in a sym-pactus, 
there exists a case that a state node $k\in\mathcal{V}_{i}^{S}$ in $\mathcal{G}_{i}$ guarantees 
the existence of a dedicated node of a state node $l\in\mathcal{V}_{j}^{S}$
in another component $\mathcal{G}_{j}$, which is adjacent to the node of $k$, i.e., $l\in\mathcal{N}_{k}$,
we call such state nodes \textit{component input nodes}.
%Thus, we re-defined the set of input nodes as the union of the set of extenal input nodes and the set of component input nodes from a perspective of dedicated node
Thus, from a component point of veiw, the input nodes can be classified as the external input nodes and the component input nodes.

\begin{definition}\label{external_input_node} 
(External input node) A set of external input nodes in $\mathcal{G}_{i}$ is symbolically written as $\mathcal{V}_{i}^{IE}$. 
If a node $k\in\mathcal{V}_{i}^{I}$ guarantees a dedicated node of $l\in\mathcal{V}^{S}_{i}$ with a directed edge $(k,l)\in\mathcal{E}$,
the node $k$ is called an external input node of $\mathcal{G}_i$, i.e., $k\in\mathcal{V}^{IE}_{i}$.
\end{definition}

%\begin{definition}\label{component_input_node} 
%(Component input node) A sets of component input nodes in $\mathcal{G}_{i}$ and $\mathcal{G}_{j}$ 
%are symbolically written as $\mathcal{V}_{i}^{IC}$ and $\mathcal{V}_{j}^{IC}$, respectively. 
%Suppose that the graphs $\mathcal{G}_{i}$ and $\mathcal{G}_{j}$ 
%If a node $k\in\mathcal{V}_{i}^{S}$ guarantees a dedicated node of $l\in\mathcal{V}^{S}_{j}$ with an undirected edge $(k,l)$,
%the node $k$ is called a component input node in $\mathcal{G}_i$, i.e., $k\in\mathcal{V}^{IC}_{i}$.
%\end{definition}

\begin{definition}\label{component_input_node} 
(Component input node) A set of component input nodes in $\mathcal{G}_{i}$ is symbolically written as $\mathcal{V}_{i}^{IC}$. 
Consider a graph $\mathcal{G}(T)=\mathcal{G}_{i}\cup\mathcal{G}_{ij}^{S}\cup\mathcal{G}_{j}$.
If a node $k\in\mathcal{V}_{i}^{S}$ guarantees a dedicated node of $l\in\mathcal{V}^{S}_{j}$ with an undirected edge $(k,l)\in\mathcal{E}_{ij}^{S}$,
the node $k$ is called a component input node of $\mathcal{G}_j$, i.e., $k\in\mathcal{V}^{IC}_{j}$.
\end{definition}

\begin{figure}[]
\centering
\subfigure[$\mathcal{G}(T)=\mathcal{G}_{1}\cup\mathcal{G}_{12}\cup\mathcal{G}_{2}$ with $|\mathcal{V}^{IE}|=2$]{
\begin{tikzpicture}[scale=0.5]
\node[] at (2.9,4.7) {\scriptsize$\mathcal{G}_1$};
\node[] at (7,6) {\scriptsize$\mathcal{G}_2$};

\node[red] at (4.7,4.2) {\scriptsize$\mathcal{G}_{12}^{S}$};

\node[place, black] (node1) at (1,5) [label=below:\scriptsize$1$] {};
\node[place, red] (node2) at (3,4) [label=below:\scriptsize$2$] {};
\node[place, red] (node3) at (5,3) [label=below:\scriptsize$3$] {};
\node[place, red] (node4) at (7,4) [label=below:\scriptsize$4$] {};
\node[place, black] (node5) at (5,5) [label=right:\scriptsize$5$] {};
\node[place, red] (node6) at (5,7) [label=above:\scriptsize$6$] {};

\node[place, black] (node7) at (7,8) [label=above:\scriptsize$7$] {};
\node[place, black] (node8) at (9,7) [label=above:\scriptsize$8$] {};
\node[place, black] (node9) at (9,5) [label=below:\scriptsize$9$] {};
%\draw[black,dashed] (3,0.5) circle (2);

\draw (node1) [line width=0.5pt] -- node [left] {} (node2);
\draw (node2) [line width=0.5pt] -- node [left] {} (node3);
\draw (node2) [red, dashed, line width=0.5pt] -- node [below] {} (node6);
\draw (node3) [red, dashed, line width=0.5pt] -- node [below] {} (node4);
\draw (node4) [line width=0.5pt] -- node [left] {} (node5);
\draw (node5) [line width=0.5pt] -- node [left] {} (node6);
\draw (node6) [line width=0.5pt] -- node [left] {} (node7);
\draw (node7) [line width=0.5pt] -- node [left] {} (node8);
\draw (node8) [line width=0.5pt] -- node [left] {} (node9);
\draw (node4) [line width=0.5pt] -- node [left] {} (node9);

\node[place, circle] (node17) at (-0.5,4) [label=below:\scriptsize$u_{1}$] {}; 
\node[place, circle] (node18) at (3.5,2) [label=below:\scriptsize$u_{2}$] {}; 
\draw (node17) [-latex, line width=0.5pt] -- node [right] {} (node1);
\draw (node18) [-latex, line width=0.5pt] -- node [right] {} (node3);
\end{tikzpicture}
}
\subfigure[$\mathcal{G}_{2}$ with $|\mathcal{V}^{IC}_{2}|=2$]{
\begin{tikzpicture}[scale=0.43]
\node[] at (7,6) {\scriptsize$\mathcal{G}_2$};

\node[place, black] (node4) at (7,4) [label=below:\scriptsize$4$] {};
\node[place, black] (node5) at (5,5) [label=below:\scriptsize$5$] {};
\node[place, black] (node6) at (5,7) [label=above:\scriptsize$6$] {};

\node[place, black] (node7) at (7,8) [label=above:\scriptsize$7$] {};
\node[place, black] (node8) at (9,7) [label=above:\scriptsize$8$] {};
\node[place, black] (node9) at (9,5) [label=below:\scriptsize$9$] {};
%\draw[black,dashed] (3,0.5) circle (2);

\draw (node4) [line width=0.5pt] -- node [left] {} (node5);
\draw (node5) [line width=0.5pt] -- node [left] {} (node6);
\draw (node6) [line width=0.5pt] -- node [left] {} (node7);
\draw (node7) [line width=0.5pt] -- node [left] {} (node8);
\draw (node8) [line width=0.5pt] -- node [left] {} (node9);
\draw (node4) [line width=0.5pt] -- node [left] {} (node9);

\node[place, circle] (node17) at (5,3) [label=below:\scriptsize$3$] {}; 
\node[place, circle] (node18) at (3,4) [label=below:\scriptsize$2$] {}; 
\draw (node4) [line width=0.5pt] -- node [left] {} (node17);
\draw (node6) [line width=0.5pt] -- node [left] {} (node18);
\end{tikzpicture}
}
\caption{Example of the SSC state nodes and the component input nodes, i.e., $\mathcal{V}^{IC}_{2}=\{ 2,3 \}$ and $\mathcal{V}^{SC}=\{ 1,2,3,4,6 \}$.}
\label{network_ex_TCinput}
\end{figure}
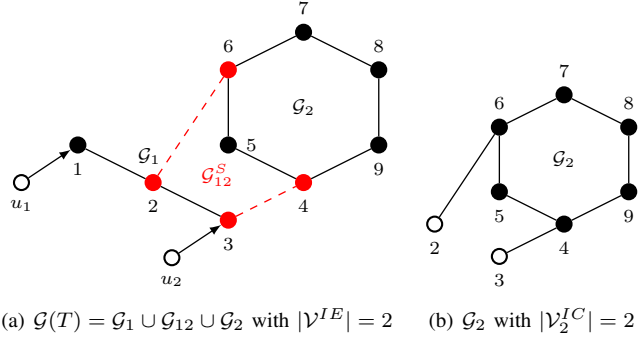

In this section, the set of input nodes $\mathcal{V}^{I}$ is re-defined as 
a union of the set of component input nodes and the set of external input nodes, 
i.e., $\mathcal{V}^{I}=\mathcal{V}^{IE}\cup\mathcal{V}^{IC}$ satisfying $\mathcal{V}^{IE}\cap\mathcal{V}^{IC}=\emptyset$.
For a sym-pactus, the following lemma provides the condition for having component input nodes.

\begin{lemma}\label{lemma_TCC}
Consider a graph $\mathcal{G}(T)=\mathcal{G}^S\cup\mathcal{G}^I$ 
with a sym-pactus $\mathcal{G}^{S}={\bigcup}^{m}_{i=1}(\mathcal{G}_{i}^{S}\cup\mathcal{G}_{ij}^{S})$ for $i,j \in \{1,...,m\}$ and $|i-j|=1$.
%Let us consider two adjacency bridge nodes $k,l\in\mathcal{V}_{ij}^{S}$ satisfying $k\in\mathcal{V}_{i}^{S}$ and $l\in\mathcal{V}_{j}^{S}$. 
The state node in $\mathcal{V}_{i}^{S}\cap\mathcal{V}_{ij}^{S}$, are the component input nodes of $\mathcal{G}_{j}$, 
i.e., $\mathcal{V}_{j}^{IC}=\mathcal{V}_{i}^{S}\cap\mathcal{V}_{ij}^{S}$,
if and only if $\mathcal{G}_{i}\cup\mathcal{G}_{ij}^{S}$ is SSC.
\end{lemma}

\begin{proof}
Note that if $\mathcal{G}_{i}\cup\mathcal{G}_{ij}^{S}$ is SSC, then  $\mathcal{G}_{i}$ is also SSC while the converse is not satisfied.
For \textit{if} condition,
let us suppose that $\mathcal{G}_{i}\cup\mathcal{G}_{ij}^{S}$ is SSC.
It follows from \textit{Corollary~\ref{corollary_3}} that all nodes in $\mathcal{V}_{i}^{S}\cup\mathcal{V}_{ij}^{S}$ are SSC state nodes.
Obviously, all the pairs of bridge nodes $k,l\in\mathcal{V}_{ij}^{S}$ satisfying $k\in\mathcal{N}_{l}$, are also SSC state nodes.
It means that the node $k\in\mathcal{V}^{S}_{i}$ guarantees the existence of a dedicated node of the node $l\in\mathcal{V}^{S}_{j}$.
Therefore, according to \textit{Definition~\ref{component_input_node}}, the bridge nodes in $\mathcal{V}_{i}^{S}\cap\mathcal{V}_{ij}^{S}$ are the component input nodes of $\mathcal{G}_{j}$,
i.e., $\mathcal{V}_{j}^{IC}=\mathcal{V}_{i}^{S}\cap\mathcal{V}_{ij}^{S}$.

For \textit{only if} condition, %consider that $\mathcal{G}_{i}\cup\mathcal{G}_{ij}^{S}$ is not TC.
suppose that $\mathcal{G}_{i}\cup\mathcal{G}_{ij}^{S}$ is not SSC.
Then, since it is trivial when $\mathcal{G}_{i}$ is not SSC, consider the graph $\mathcal{G}_{i}$ is SSC.
Then, the bridge nodes in $\mathcal{V}_{i}^{S}\cap\mathcal{V}_{ij}^{S}$ are SSC state node.
However, the bridge nodes in $\mathcal{V}_{j}^{S}\cap\mathcal{V}_{ij}^{S}$ are not SSC state node.
Therefore, there exists a case without a dedicated node when $\alpha\subseteq\mathcal{V}^{S}$ includes a node in $\mathcal{V}_{j}^{S}\cap\mathcal{V}_{ij}^{S}$.
\end{proof}

Note that if \textit{Lemma~\ref{lemma_TCC}} is satisfied, the set of component input nodes is defined as $\mathcal{V}^{IC}_{j}=\mathcal{V}^{S}_{i}\cap\mathcal{V}^{S}_{ij}$.
For example, in the graph $\mathcal{G}(T)$ depicted in Fig.~\ref{network_ex_TCinput}(a), 
consider the SSC subgraph $\mathcal{G}_{1}\cup\mathcal{G}_{12}^{S}$ satisfying \textit{Lemma~\ref{lemma_tree}}.
Then, according to \textit{Lemma~\ref{lemma_TCC}}, the nodes $2,3\in\mathcal{V}_{1}^{S}\cap\mathcal{V}_{12}^{S}$ are 
the component input nodes of $\mathcal{G}_{2}$, i.e., $\mathcal{V}_{2}^{IC}=\{2,3\}$.
From the viewpoint of $\mathcal{G}_{2}$, the graph depicted in Fig.~\ref{network_ex_TCinput}(a) can be expressed as shown in Fig.~\ref{network_ex_TCinput}(b).
Note that the role of external input nodes and component input nodes is equivalent from a perspective of guaranteeing a dedicated node of a state node.
Now, let us suppose that the graph $\mathcal{G}_{2}$ in Fig.~\ref{network_ex_TCinput}(b) satisfies \textit{Lemma~\ref{lemma_cycle}}
by an additional properly located external input node connected to one of the nodes $5,7$, and $9$.
Then, all state nodes in $\mathcal{G}(T)$ become SSC state nodes, it follows from \textit{Corollary~\ref{corollary_3}} that $\mathcal{G}(T)$ is SSC.
In this manner, for a general type of sym-pactus state graph, based on the aforementioned component input nodes,
a SSC graph with the minimum number of external input nodes can be designed by adding additional properly located external input nodes.
In this paper, for a graph $\mathcal{G}(T)=\mathcal{G}^{S}\cup\mathcal{G}^{I}$, the minimum number of external input nodes in $\mathcal{V}^{I}$ 
for the strong sign controllability is symbolically written as $\min{|\mathcal{V}^{IE}|}$.
For a certain type of sym-pactus,
the following theorem provides a merging algorithm of polynomial computational complexity
that can uniquely determine the minimum number of external input nodes while maintaining the strong sign controllability.

\begin{theorem} \label{theorem_sym_pactus_nc_condition}
Let a state graph $\mathcal{G}^{S}$ be a sym-pactus with $\mathcal{E}_{ij}^{S}$ satisfying $|\mathcal{E}_{ij}^{S}|=1$ 
and assume that $\mathcal{G}_{i}^{S}\cup\mathcal{G}_{ij}^{S}$ does not contain a tree for $i,j\in\{ 1,...,m\}$ and $|i-j|=1$. 
Then, the output graph $\mathcal{G}(T)^{out}$ in \textit{Algorithm~\ref{algorithm_graph_merging_min}} is SSC satisfying $\min{|\mathcal{V}^{IE}|=c+1}$, 
where $c$ is the number of sym-cycle components.
\end{theorem}
\begin{proof}
%\textit{Theorem~\ref{theorem_sym_pactus_nc_condition}} can be easily proved using \textit{Lemma~\ref{lemma_path}}, \textit{Lemma~\ref{lemma_cycle}},  \textit{Lemma~\ref{lemma_tree}}, 
%and \textit{Algorithm~\ref{algorithm_graph_merging_min}}.
For \textit{if} condition, let us classify the components $\mathcal{G}_{i}^{S},i\in\{1,...,m\}$ in sym-pactus $\mathcal{G}^{S}$  as $p$ sym-path and $c$ sym-cycle components satisfying $m=p+c$. 
For $\mathcal{G}_{1}$, since the component input nodes can not exist in $\mathcal{G}_{1}$, i.e., $\mathcal{V}^{IC}_{1}=\emptyset$, it needs to be SSC with only external input nodes.
Then, let $\mathcal{G}_{1}$ be a sym-path,
according to \textit{Lemma~\ref{lemma_path}}, a properly located external input node is needed to be SSC, i.e., $|\mathcal{V}_{1}^{IE}|=1$. 
Conversely, if $\mathcal{G}_{1}$ is a sym-cycle, two properly located external input nodes are required to satisfy \textit{Lemma~\ref{lemma_cycle}}, i.e., $|\mathcal{V}_{1}^{IE}|=2$. 
Then, all state nodes in $\mathcal{G}_{1}$ are the SSC state nodes, i.e., $\mathcal{V}_{1}^{S}=\mathcal{V}_{1}^{SC}$.
Since the assumption that $\mathcal{G}^{S}_{j-1}\cup\mathcal{G}_{(j-1)j}^{S},j\in\{2,...,m\}$, does not contain a tree and the condition of $|\mathcal{E}_{ij}^{S}|=1$,
each component $\mathcal{G}_{j},j\in\{2,...,m\}$ has a component input node $\mathcal{V}_{j}^{IC}=\mathcal{V}_{j-1}^{S}\cap\mathcal{V}_{(j-1)j}^{S}$.
Thus, if $\mathcal{G}_{j}$ is a sym-path, it satisfies \textit{Lemma~\ref{lemma_path}} without additional external input nodes
because $\mathcal{G}_{j}\cup\mathcal{G}_{(j-1)j}$ has a properly located component input node.
Therefore, $\mathcal{G}_{j}$ needs an additional external input node 
only when $\mathcal{G}_{j}$ is a sym-cycle, i.e., $|\mathcal{V}_{j}^{IE}|=|\mathcal{V}_{j}^{IC}|=1$.
After adding an external input node, all state nodes in $\mathcal{G}_{j}$ are the SSC state nodes, i.e., $\mathcal{V}_{j}^{S}=\mathcal{V}_{j}^{SC}$ for $j\in\{2,...,m\}$.
According to \textit{Theorem~\ref{theorem_TCC}}, the graph $\mathcal{G}$ is SSC 
since the union set of the SSC state nodes of each component $\mathcal{G}_{i},i\in\{1,...,m\}$, satisfies ${\bigcup}^{m}_{i=1}\mathcal{V}_{i}^{SC}=\mathcal{V}^{S}$.
Furthermore, the SSC graph $\mathcal{G}(T)$ has the minimum number of external input nodes 
since each component (sym-path or sym-cycle) has the minimum number of input nodes containing all possible component input nodes
$|\mathcal{V}^{I}|=|\mathcal{V}^{IE}|+|\mathcal{V}^{IC}|$ to satisfy \textit{Lemma~\ref{lemma_path}} and \textit{Lemma~\ref{lemma_cycle}}.
Therefore, the minimum number of external input nodes for strong sign controllability of the graph $\mathcal{G}(T)$ is $c+1$.

For \textit{only if} condition, 
in \textit{Algorithm~\ref{algorithm_graph_merging_min}}, 
each component $\mathcal{G}_{i}, i\in\{1,...,m\}$ has the minimum number of external input nodes with all possible component input nodes
based on \textit{Lemma~\ref{lemma_path}}, \textit{Lemma~\ref{lemma_cycle}}, and \textit{Lemma~\ref{lemma_TCC}}.
Therefore, if more than necessary external input nodes are added in one of each step, 
the total number of external input nodes satisfies ${\bigcup}^{m}_{i=1}\min{|\mathcal{V}^{IE}_{i}|}>c+1$.
Note that the locations of external input nodes are not unique.
% 알고리즘같은경우에는 대우증명이 가능함?? 케이스별로 나눠서해야할것같은데 너무 길어짐..
%If a sym-path component $\mathcal{G}_{k}, k\in\{1,...,m\}$ has more input nodes than necessary to be TC
%\textcolor{red}{START FROM HERE} \newline
%in case of $\mathcal{G}_{1}$, the minimum number of external input nodes corresponding to 
%graph types is $|\mathcal{V}^{IE}_{1}|=1$ (sym-path) or $|\mathcal{V}^{IE}_{1}|=2$ (sym-cycle).
%For $\mathcal{G}_{j},j\in\{2,...,m\}$,
%let us consider a sym-cycle $\mathcal{G}_{j}$ has $\mathcal{V}_{j}^{IE}$ satisfying $|\mathcal{V}_{j}^{IE}|> 1$.
%Then, since $\mathcal{G}_{j}$ has more input nodes than necessary to be TC, i.e., $|\mathcal{V}_{j}^{I}|>2$ with $|\mathcal{V}_{j}^{IE}|>1$ and $|\mathcal{V}_{j}^{IC}|=1$.
%Therefore, the number of external input nodes becomes $|\mathcal{V}^{IE}|>c+1$.
\end{proof}
\begin{algorithm}
\caption{}\label{algorithm_graph_merging_min}
\begin{algorithmic}[1]
\Procedure{}{}
\State $i=0, c=0, p=0$
 \For{$i=i+1$} 
  \Switch{graph type of $\mathcal{G}_{i}^{S}$}
    \Case{\textit{sym-path}}
\State $p:=p+1$
\If { $i=1$} 
\State add $\mathcal{V}^{IE}_{i},|\mathcal{V}^{IE}_{i}|=1$ to satisfy \textit{Lemma~\ref{lemma_path}}
\Else ~break
\EndIf
\State ~$\mathcal{G}_{i}:=\mathcal{G}_{i}^{S}\cup\mathcal{G}_{i}^{I}$
    \EndCase
    \Case{\textit{sym-cycle}}
\State $c:=c+1$
\If { $i=1$} 
\State add $\mathcal{V}^{IE}_{i},|\mathcal{V}^{IE}_{i}|=2$ to satisfy \textit{Lemma~\ref{lemma_cycle}}
\Else  
\State add $\mathcal{V}^{IE}_{i},|\mathcal{V}^{IE}_{i}|=1$ to satisfy \textit{Lemma~\ref{lemma_cycle}}
\EndIf
\State ~$\mathcal{G}_{i}:=\mathcal{G}_{i}^{S}\cup\mathcal{G}_{i}^{I}$
    \EndCase
  \EndSwitch
\If { $i < m$}
    \State determine $\mathcal{V}_{i+1}^{IC} :=\mathcal{V}_{i}^{S}\cap\mathcal{V}_{i(i+1)}^{S}$
    \Else %~check $m=p+c$
    \State check $m=p+c$
     \textit{goto} \textit{end}
    \EndIf
  \EndFor   \Comment{End for $i$}
\BState \emph{end}:
\State \textit{output} 
$\mathcal{G}(T)^{out}=\bigcup^{m}_{i=1}\mathcal{G}_{i}$ with $min|\mathcal{V}^{IE}|=c+1$
\EndProcedure
\end{algorithmic}
\end{algorithm}

\begin{corollary} \label{corollary_sym_pactus_nc_condition}
Let a state graph $\mathcal{G}^{S}$ be a sym-pactus with $\mathcal{E}_{ij}^{S}$ satisfying $|\mathcal{E}_{ij}^{S}|=1$ 
and assume that each component $\mathcal{G}_{i}^{S}$ is a sym-cycle. 
Then, the output graph $\mathcal{G}(T)^{out}$ in \textit{Algorithm~\ref{algorithm_graph_merging_min}} is SSC satisfying $\min{|\mathcal{V}^{IE}|=c+1}$.
\end{corollary}

As an example of \textit{Algorithm~\ref{algorithm_graph_merging_min}}, 
the state graph $\mathcal{G}^{S}$ depicted in Fig.~\ref{network_ex_th2}(a) is a sym-pactus with $\mathcal{E}_{ij}^{S}$ satisfying $|\mathcal{E}_{ij}^{S}|=1$ for $i,j\in\{1,2,3,4\}$ and $|i-j|=1$. 
Thus, $\mathcal{G}$ can be induced by $4$-disjoint components $\mathcal{G}_{i}, i\in\{1,2,3,4\}$ with a bridge edge. 
For $\mathcal{G}_{1}$, since $\mathcal{G}_{1}^{S}$ is a sym-path, 
it needs an additional external input node $u_{1}\in\mathcal{V}^{IE}_{1}$ connected to the terminal state node $1$ to satisfy \textit{Lemma~\ref{lemma_path}}. 
Next, according to \textit{Lemma~\ref{lemma_TCC}}, $\mathcal{G}_{2}$ has a component input node $3\in\mathcal{V}^{IC}_{2}$ from $\mathcal{G}_{1}$.
Thus, the sym-cycle $\mathcal{G}_{2}^{S}$ requires an additional external input node $u_{2}\in\mathcal{V}^{IE}_{2}$ connected to the node $5$ 
to satisfy \textit{Lemma~\ref{lemma_cycle}} as shown in Fig.~\ref{network_ex_th2}(c).
In the same way, $\mathcal{G}_{3}$  has a component input node $9\in\mathcal{V}^{IC}_{3}$ from $\mathcal{G}_{2}$, 
so an additional external input node $u_{3}\in\mathcal{V}^{IE}_{3}$ connected to the node 10 is necessary to make $\mathcal{G}_{3}$ SSC as shown in Fig.~\ref{network_ex_th2}(c).
Also, $\mathcal{G}_{4}$ has a component input node $11\in\mathcal{V}^{IC}_{4}$ from $\mathcal{G}_{3}$,
thus, $\mathcal{G}_{4}$ requires an additional external input node $u_{4}\in\mathcal{V}^{IE}_{4}$ at one of the nodes 13 and 15 (we choose node 13) as shown in Fig.~\ref{network_ex_th2}(c).
Therefore, the minimum number of external input nodes is $4$ to satisfy \textit{Theorem~\ref{theorem_TCC}} as shown in Fig.~\ref{network_ex_th2}(c),
i.e., $\mathcal{V}^{IE}={\bigcup}^{4}_{i=1}\mathcal{V}_{i}^{IE}=\{ u_{1},u_{2},u_{3},u_{4}\}$.

As an extension of \textit{Algorithm~\ref{algorithm_graph_merging_min}},
we propose an algorithm, 
which is applicable to a general type of a sym-pactus.
The algorithm starts by a \textit{decomposition process}, which is a decomposition of a sym-pactus into basic components.
In \textit{decomposition process}, the graph $\bar{\mathcal{G}}_{i}^{S}$ denotes a union state graphs of $i$-th component and its bridge graph, 
i.e., $\bar{\mathcal{G}}_{i}^{S} := \mathcal{G}_{i}^{S}\cup\mathcal{G}_{i(i+1)}^{S}$.
%\textcolor{red}{
%Also, we classify graph types of  $\bar{\mathcal{G}}_{i}^{S}$ as \textit{path-type}, \textit{tree-type}, and \textit{cycle-type}.
%$\bar{\mathcal{G}}_{i}^{S}$ is \textit{path-type} if $\bar{\mathcal{G}}_{i}^{S}$ is a path graph, whereas
%$\bar{\mathcal{G}}_{i}^{S}$ is \textit{tree-type} if $\bar{\mathcal{G}}_{i}^{S}$ is a tree graph.
%Otherwise, that is, if $\bar{\mathcal{G}}_{i}^{S}$ contains a cycle graph, then $\bar{\mathcal{G}}_{i}^{S}$ is \textit{cycle-type}.
%}
Also, we classify graph types of  $\bar{\mathcal{G}}_{i}^{S}$ as \textit{path-type}, \textit{tree-type}, and \textit{cycle-type}.
The graph $\bar{\mathcal{G}}_{i}^{S}$ is \textit{path-type} and \textit{tree-type} if $\bar{\mathcal{G}}_{i}^{S}$ is a path and a tree graph, respectively.
Otherwise, if $\bar{\mathcal{G}}_{i}^{S}$ contains a cycle graph, then $\bar{\mathcal{G}}_{i}^{S}$ is \textit{cycle-type}.
In particular, let us consider a graph $\mathcal{G}_{i}\cup\mathcal{G}_{ij}^{S}$ with $\mathcal{V}^{I}$ satisfying $|\mathcal{V}^{I}|=m\geq 2$,
where $\mathcal{G}_{i}^{S}$ is a sym-cycle.
Then, it follows from \textit{Corollary~\ref{corollary_2}} that if $\mathcal{G}_{i}\cup\mathcal{G}_{ij}^{S}$ can be induced by $m$-disjoint sym-paths satisfying \textit{Lemma~\ref{lemma_path}}, 
the graph $\mathcal{G}_{i}\cup\mathcal{G}_{ij}^{S}$ is SSC.
In \textit{merging process}, the minimum number of external input nodes of each component is sequentially added at the proper locations step-by-step 
based on \textit{Lemma~\ref{lemma_path}}, \textit{Lemma~\ref{lemma_tree}}, and \textit{Lemma~\ref{lemma_cycle}}.
Then, according to \textit{Lemma~\ref{lemma_TCC}}, the component input nodes for each step are determined as $\mathcal{V}_{i+1}^{IC} :=\mathcal{V}_{i}^{S}\cap\mathcal{V}^{S}_{i(i+1)}$.
Based on the aforementioned statements, the graph merging algorithm for a general type of a sym-pactus can be produced as \textit{Algorithm~\ref{graph_merging}}.

\begin{remark}\label{remark_complexity}
The \textit{Algorithm~\ref{algorithm_graph_merging_min}} and \textit{Algorithm~\ref{graph_merging}} provide 
a graph theoretic method to ensure the strong sign controllability of $\mathcal{G}(T)=(\mathcal{V},\mathcal{E})$ with the minimum number of external input nodes.
Similarly, for structural controllability, a method of finding the minimum number of leaders (inputs) has been developed in \cite{guan2021structural}, and
the complexity of the Theorem 2 in \cite{guan2021structural} is $\mathcal{O}(|\mathcal{V}|+|\mathcal{E}|)$.
However, the complexity of the \textit{Algorithm~\ref{algorithm_graph_merging_min}} and \textit{Algorithm~\ref{graph_merging}} is $\mathcal{O}(m)$,
where $m$ is the number of disjoint components constituting a sym-pactus. 
Note that $m$ is smaller than $|\mathcal{V}|$ because the simplest structure of a sym-pactus must have a sym-cycle type of disjoint component containing at least three state nodes (see \textit{Definition~\ref{definition_sym_cycle}}).
\end{remark}

the exact number of calculations

\begin{algorithm}
\caption{}\label{graph_merging}
\begin{algorithmic}[1]
\Procedure{}{}
\BState \emph{decomposition}:
\State a sym-pactus $\mathcal{G}^{S}$ := ${\bigcup}^{m}_{i=1}(\mathcal{G}^{S}_{i}\cup\mathcal{G}_{i(i+1)}^{S})$,  $\mathcal{V}^{IE}=\emptyset$ 
\State $i=0$
 \For{$i=i+1$} 
     \If { $i < m$}
    \State $\bar{\mathcal{G}}_{i}^{S} := \mathcal{G}_{i}^{S}\cup\mathcal{G}_{i(i+1)}^{S}$
%    \ElsIf { $i = m$}
%    \State $\bar{\mathcal{G}}_{i} := \mathcal{G}_{i}^{S}$
    \Else ~$\bar{\mathcal{G}}_{i}^{S} := \mathcal{G}_{i}^{S}$
    \State 
     \textit{goto} \textit{merging}
    \EndIf
%\State determine the type of $\bar{\mathcal{G}}_{i}$
  \EndFor   \Comment{End for $i$}
\BState \emph{merging}:
\State $i=0$
 \For{$i=i+1$} 
\If {$\bar{\mathcal{G}}_{i}^{S}$ is \textit{path-type}} k=1
\ElsIf {$\bar{\mathcal{G}}_{i}^{S}$ is \textit{tree-type}} k=2
\Else ~k=3
\EndIf
\If { $\bar{\mathcal{G}}_{i}^{S}$ is SSC} break;
\Else   ~add $\mathcal{V}^{IE}_{i}$ to satisfy \textit{Lemma k}
\State ~update $\bar{\mathcal{G}}_{i}^{I}:=\mathcal{G}_{i}^{I}$%}
\EndIf
\State ~$\bar{\mathcal{G}}_{i}:=\bar{\mathcal{G}}_{i}^{S}\cup\bar{\mathcal{G}}_{i}^{I}$
%    \EndCase
%    \Case{\textit{tree-type}}
%\If { $\bar{\mathcal{G}}_{i}^{S}$ is TC} break;
%\Else   ~add $\mathcal{V}^{IE}_{i}$ to satisfy \textit{Lemma~\ref{lemma_tree}}
%\State ~update $\bar{\mathcal{G}}_{i}^{I}:=\mathcal{G}_{i}^{I}$ 
%\EndIf
%\State ~$\bar{\mathcal{G}}_{i}:=\bar{\mathcal{G}}_{i}^{S}\cup\bar{\mathcal{G}}_{i}^{I}$
%    \EndCase
%    \Case{\textit{cycle-type}}
%\If { $\bar{\mathcal{G}}_{i}^{S}$ is TC} break;
%\Else   ~add $\mathcal{V}^{IE}_{i}$ to satisfy \textit{Lemma~\ref{lemma_cycle}}
%\State ~update $\bar{\mathcal{G}}_{i}^{I}:=\mathcal{G}_{i}^{I}$ 
%\EndIf
%\State ~$\bar{\mathcal{G}}_{i}:=\bar{\mathcal{G}}_{i}^{S}\cup\bar{\mathcal{G}}_{i}^{I}$
%    \EndCase
%  \EndSwitch
\If { $i < m$}
    \State determine $\mathcal{V}_{i+1}^{IC} :=\mathcal{V}_{i}^{S}\cap\mathcal{V}^{S}_{i(i+1)}$
    \Else 
    \State 
     \textit{goto} \textit{end}
    \EndIf
  \EndFor   \Comment{End for $i$}
\BState \emph{end}:
\State \textit{output} $\mathcal{G}^{out}=\bigcup^{m}_{i=1}\bar{\mathcal{G}}_{i}$  with $min|\mathcal{V}^{IE}|$
\EndProcedure
\end{algorithmic}
\end{algorithm}

\section{Examples} \label{sec_topo_ex}
This section first introduces a topological example of the graph merging algorithm based on \textit{Algorithm~\ref{graph_merging}}.
Then, the strong sign controllability of the output graph of \textit{Algorithm~\ref{graph_merging}} will be verified from a numerical approach.

\subsection{Graph merging algorithm for a sym-pactus}
Consider the sym-pactus $\mathcal{G}^{S}=(\mathcal{V}^{S},\mathcal{E}^{S})$ depicted in Fig.~\ref{network_ex_final} consisting of four disjoint components and three bridge graphs.
In \textit{decomposition process}, the sym-pactus $\mathcal{G}^{S}$ is decomposed into disjoint components and bridge graphs.
For simplicity, each union graph of the disjoint component and the bridge graph is defined as $\bar{\mathcal{G}}_{i}^{S}, i\in\{1,2,3,4\}$.
The set of state nodes for each $\bar{\mathcal{G}}_{i}^{S}$ and corresponding graph types are given by:
\newline \newline
$\mathcal{V}^{S}_{1}=\{ 1,2,3 \}$\,\,\,\,\,\,\,\,\,\,\,\,\,\,\,\,\,\,\,\,\,\,$\mathcal{V}^{S}_{12}=\{ 2,3,4,6 \}$,
\newline
$\mathcal{V}^{S}_{2}=\{ 4,5,6,7,8,9 \}$\,\,\, $\mathcal{V}^{S}_{23}=\{ 4,9,10,12 \}$,
\newline
$\mathcal{V}^{S}_{3}=\{ 10,11,12 \}$\,\,\,\,\,\,\,\,\,\,\, $\mathcal{V}^{S}_{34}=\{ 12,13 \}$,
\newline
$\mathcal{V}^{S}_{4}=\{ 13,14,15,16 \}$\,\,\,\,$\mathcal{V}^{S}=\mathcal{V}^{S}_{1}\cup\mathcal{V}^{S}_{2}\cup\mathcal{V}^{S}_{3}\cup\mathcal{V}^{S}_{4}$
\newline
\newline
$\bar{\mathcal{G}}_{1}^{S}=\mathcal{G}_{1}^{S}\cup\mathcal{G}_{12}^{S}$ : \textit{tree-type}
\newline
$\bar{\mathcal{G}}_{2}^{S}=\mathcal{G}_{2}^{S}\cup\mathcal{G}_{23}^{S}$ : \textit{cycle-type}
\newline
$\bar{\mathcal{G}}_{3}^{S}=\mathcal{G}_{3}^{S}\cup\mathcal{G}_{34}^{S}$ : \textit{cycle-type}
\newline
$\bar{\mathcal{G}}_{4}^{S}=\mathcal{G}_{4}^{S}$ : \textit{cycle-type}
\newline

In \textit{merging process}, for each $\bar{\mathcal{G}}_{i}=\bar{\mathcal{G}}_{i}^{S}\cup\bar{\mathcal{G}}_{i}^{I}, i\in\{1,2,3,4\}$, 
additional external input nodes are added at proper locations to guarantee the existence of dedicated nodes for all $\alpha\subseteq\bar{\mathcal{V}}_{i}^{S}$,
i.e, the SSC state nodes in $\mathcal{G}$ will expand sequentially. 
As a result, the output graph $\mathcal{G}^{out}$ in \textit{Algorithm~\ref{graph_merging}} satisfies \textit{Corollary~\ref{corollary_3}}, 
which is equivalent to satisfying \textit{Theorem~\ref{theorem_Tsatsomeros}}.
For intuition, the blue marked state nodes in Fig.~\ref{network_ex_step} are the SSC state nodes.

\begin{itemize}
\item Step 1.
\end{itemize}
Since $\bar{\mathcal{G}}_{1}^{S}=\mathcal{G}_{1}^{S}\cup\mathcal{G}_{12}^{S}$ is \textit{a \textit{tree-type}},
two additional external input nodes $u_1,u_2\in\mathcal{V}_{1}^{IE}$ need to be connected at node $1$ and $3$ to satisfy \textit{Lemma~\ref{lemma_tree}}.
After that, all state nodes of $\bar{\mathcal{G}}_{1}$ are SSC state nodes as shown in Fig.~\ref{network_ex_step}(a), 
i.e., $\bar{\mathcal{V}}_{1}^{SC}=\{1,2,3,4,6\}$.

\begin{itemize}
\item Step 2.
\end{itemize}
$\bar{\mathcal{G}}_{2}^{S}=\mathcal{G}_{2}^{S}\cup\mathcal{G}_{23}^{S}$ is \textit{a \textit{cycle-type}}.
It follows from \textit{Lemma~\ref{lemma_TCC}} that the component input nodes of $\bar{\mathcal{G}}_{2}$ are $2,3 \in \bar{\mathcal{V}}_{2}^{IC}$.
Hence, one external input node $u_3\in\mathcal{V}_{2}^{IE}$ needs to be connected to node $5$, $7$ or $9$ to satisfy \textit{Lemma~\ref{lemma_cycle}} (we choose node $7$).
After that, all state nodes of $\bar{\mathcal{G}}_{2}$ are SSC state nodes as shown in Fig.~\ref{network_ex_step}(b), 
i.e., $\bar{\mathcal{V}}_{2}^{SC}=\{4,5,6,7,8,9,10,12\}$.

\begin{itemize}
\item Step 3.
\end{itemize}
Since $\bar{\mathcal{G}}_{3}^{S}=\mathcal{G}_{3}^{S}\cup\mathcal{G}_{34}^{S}$ is \textit{a \textit{cycle-type}},
two input nodes are needed to satisfy \textit{Lemma~\ref{lemma_cycle}}.
According to \textit{Lemma~\ref{lemma_TCC}}, the component input nodes of $\bar{\mathcal{G}}_{3}$ are $4,9\in\bar{\mathcal{V}}_{3}^{IC}$.
Hence, $\bar{\mathcal{G}}_{3}$ already satisfies \textit{Lemma~\ref{lemma_cycle}} without additional external input nodes
since there exists an edge $(10,12)\in\bar{\mathcal{E}}_{3}$ between the nodes $10,12\in\mathcal{N}(\bar{\mathcal{V}}_{3}^{IC})$.
Therefore, all state nodes of $\bar{\mathcal{G}}_{3}$ are SSC state nodes as shown in Fig.~\ref{network_ex_step}(c), 
i.e., $\bar{\mathcal{V}}_{3}^{SC}=\{10,11,12,13\}$.

\begin{itemize}
\item Step 4.
\end{itemize}
$\bar{\mathcal{G}}_{4}^{S}=\mathcal{G}_{4}^{S}$ is \textit{a \textit{cycle-type}},
two input nodes are needed to satisfy \textit{Lemma~\ref{lemma_cycle}}.
According to \textit{Lemma~\ref{lemma_TCC}}, the component input node of $\bar{\mathcal{G}}_{4}$ is $12\in\bar{\mathcal{V}}_{4}^{IC}$.
Thus, one external input node $u_4\in\mathcal{V}_{4}^{IE}$ needs to be added at node $14$ or $16$ to satisfy \textit{Lemma~\ref{lemma_cycle}} (we choose node $14$).
After that, all state nodes of $\bar{\mathcal{G}}_{4}$ are SSC state nodes as shown in Fig.~\ref{network_ex_step}(d),
i.e., $\bar{\mathcal{V}}_{4}^{SC}=\{13,14,15,16\}$.

Finally, all the state nodes in $\mathcal{V}^{S}$ are the SSC state nodes, i.e., $\bar{\mathcal{V}}^{SC}=\mathcal{V}^{S}$, 
where $\bar{\mathcal{V}}^{SC}={\bigcup}^{4}_{i=1}\bar{\mathcal{V}}_{i}^{SC}$,
it follows from \textit{Corollary~\ref{corollary_3}} that the output graph $\mathcal{G}(T)^{out}$ shown in Fig.~\ref{network_ex_step}(e) is SSC.
Hence, the original graph (shown in Fig.~\ref{network_ex_final}) needs at least four external input nodes $u_1,u_2,u_3,u_4\in\mathcal{V}^{IE}$ 
to satisfy the condition of strong sign controllability.
The strong sign controllability of the output graph $\mathcal{G}(T)^{out}$ is confirmed from the following numerical verification subsection.

\begin{figure}[]
\centering
\begin{tikzpicture}[scale=0.5]

\node[] at (2.9,4.7) {\scriptsize$\mathcal{G}_1$};
\node[] at (7,6) {\scriptsize$\mathcal{G}_2$};
\node[] at (10.4,3) {\scriptsize$\mathcal{G}_3$};
\node[] at (14,4.5) {\scriptsize$\mathcal{G}_4$};

\node[] at (4.7,4.2) {\scriptsize$\mathcal{G}_{12}$};
\node[] at (9,4) {\scriptsize$\mathcal{G}_{23}$};
\node[] at (12.2,4) {\scriptsize$\mathcal{G}_{34}$};

\node[place, black] (node1) at (1,5) [label=below:\scriptsize$1$] {};
\node[place, black] (node2) at (3,4) [label=below:\scriptsize$2$] {};
\node[place, black] (node3) at (5,3) [label=below:\scriptsize$3$] {};
\node[place, black] (node4) at (7,4) [label=below:\scriptsize$4$] {};
\node[place, black] (node5) at (5,5) [label=right:\scriptsize$5$] {};
\node[place, black] (node6) at (5,7) [label=above:\scriptsize$6$] {};
\node[place, black] (node7) at (7,8) [label=above:\scriptsize$7$] {};
\node[place, black] (node8) at (9,7) [label=above:\scriptsize$8$] {};
\node[place, black] (node9) at (9,5) [label=left:\scriptsize$9$] {};
\node[place, black] (node10) at (9,3) [label=below:\scriptsize$10$] {};
\node[place, black] (node11) at (11,2) [label=below:\scriptsize$11$] {};
\node[place, black] (node12) at (11,4) [label=above:\scriptsize$12$] {};
\node[place, black] (node13) at (13,5) [label=above:\scriptsize$13$] {};
\node[place, black] (node14) at (15,6) [label=above:\scriptsize$14$] {};
\node[place, black] (node15) at (15,4) [label=below:\scriptsize$15$] {};
\node[place, black] (node16) at (13,3) [label=below:\scriptsize$16$] {};

\draw (node1) [line width=0.5pt] -- node [left] {} (node2);
\draw (node2) [line width=0.5pt] -- node [left] {} (node3);
\draw (node2) [line width=0.5pt] -- node [below] {} (node6);
\draw (node3) [,line width=0.5pt] -- node [below] {} (node4);
\draw (node4) [line width=0.5pt] -- node [left] {} (node5);
\draw (node5) [line width=0.5pt] -- node [left] {} (node6);
\draw (node6) [line width=0.5pt] -- node [left] {} (node7);
\draw (node7) [line width=0.5pt] -- node [left] {} (node8);
\draw (node8) [line width=0.5pt] -- node [left] {} (node9);
\draw (node4) [line width=0.5pt] -- node [left] {} (node9);
\draw (node4) [line width=0.5pt] -- node [below] {} (node10);
\draw (node10) [line width=0.5pt] -- node [left] {} (node11);
\draw (node10) [line width=0.5pt] -- node [left] {} (node12);
\draw (node11) [line width=0.5pt] -- node [left] {} (node12);
\draw (node9) [line width=0.5pt] -- node [below] {} (node12);
\draw (node12) [line width=0.5pt] -- node [below] {} (node13);
\draw (node13) [line width=0.5pt] -- node [left] {} (node14);
\draw (node14) [line width=0.5pt] -- node [left] {} (node15);
\draw (node15) [line width=0.5pt] -- node [left] {} (node16);
\draw (node13) [line width=0.5pt] -- node [left] {} (node16);
 
%\node[place, circle] (node17) at (-0.5,4) [label=below:\scriptsize$u_{1}$] {}; 
%\node[place, circle] (node18) at (3.5,2) [label=below:\scriptsize$u_{2}$] {}; 
%\node[place, circle] (node20) at (6,9) [label=below:\scriptsize$u_{3}$] {}; 
%\node[place, circle] (node19) at (13.5,7) [label=below:\scriptsize$u_{4}$] {}; 
%\draw (node1) [-latex, line width=0.5pt] -- node [right] {} (node17);
%\draw (node3) [-latex, line width=0.5pt] -- node [right] {} (node18);
%\draw (node7) [-latex, line width=0.5pt] -- node [right] {} (node20);
%\draw (node14) [-latex, line width=0.5pt] -- node [right] {} (node19);

\end{tikzpicture}
\caption{A sym-pactus $\mathcal{G}^{S}={\bigcup}^{4}_{i=1}(\mathcal{G}_{i}^{S}\cup\mathcal{G}_{ij}^{S})$ for $|i-j|=1$.}
\label{network_ex_final}
%\end{figure}

%\begin{figure}[]
\subfigure[Step 1 : $\bar{\mathcal{G}}_{1}=\mathcal{G}_{1}\cup\mathcal{G}_{12}$]{
\begin{tikzpicture}[scale=0.5]

\node[] at (2.9,4.7) {\scriptsize$\mathcal{G}_1$};
\node[] at (7,6) {\scriptsize$\mathcal{G}_2$};
\node[] at (10.4,3) {\scriptsize$\mathcal{G}_3$};
\node[] at (14,4.5) {\scriptsize$\mathcal{G}_4$};

\node[] at (4.7,4.2) {\scriptsize$\mathcal{G}_{12}$};
\node[] at (9,4) {\scriptsize$\mathcal{G}_{23}$};
\node[] at (12.2,4) {\scriptsize$\mathcal{G}_{34}$};

\node[place, blue] (node1) at (1,5) [label=below:\scriptsize$1$] {};
\node[place, blue] (node2) at (3,4) [label=below:\scriptsize$2$] {};
\node[place, blue] (node3) at (5,3) [label=below:\scriptsize$3$] {};
\node[place, blue] (node4) at (7,4) [label=below:\scriptsize$4$] {};
\node[place, black] (node5) at (5,5) [label=right:\scriptsize$5$] {};
\node[place, blue] (node6) at (5,7) [label=above:\scriptsize$6$] {};
\node[place, black] (node7) at (7,8) [label=above:\scriptsize$7$] {};
\node[place, black] (node8) at (9,7) [label=above:\scriptsize$8$] {};
\node[place, black] (node9) at (9,5) [label=left:\scriptsize$9$] {};
\node[place, black] (node10) at (9,3) [label=below:\scriptsize$10$] {};
\node[place, black] (node11) at (11,2) [label=below:\scriptsize$11$] {};
\node[place, black] (node12) at (11,4) [label=above:\scriptsize$12$] {};
\node[place, black] (node13) at (13,5) [label=above:\scriptsize$13$] {};
\node[place, black] (node14) at (15,6) [label=above:\scriptsize$14$] {};
\node[place, black] (node15) at (15,4) [label=below:\scriptsize$15$] {};
\node[place, black] (node16) at (13,3) [label=below:\scriptsize$16$] {};

\draw (node1) [line width=0.5pt] -- node [left] {} (node2);
\draw (node2) [line width=0.5pt] -- node [left] {} (node3);
\draw (node2) [line width=0.5pt] -- node [below] {} (node6);
\draw (node3) [line width=0.5pt] -- node [below] {} (node4);
\draw (node4) [line width=0.5pt] -- node [left] {} (node5);
\draw (node5) [line width=0.5pt] -- node [left] {} (node6);
\draw (node6) [line width=0.5pt] -- node [left] {} (node7);
\draw (node7) [line width=0.5pt] -- node [left] {} (node8);
\draw (node8) [line width=0.5pt] -- node [left] {} (node9);
\draw (node4) [line width=0.5pt] -- node [left] {} (node9);
\draw (node4) [line width=0.5pt] -- node [below] {} (node10);
\draw (node10) [line width=0.5pt] -- node [left] {} (node11);
\draw (node10) [line width=0.5pt] -- node [left] {} (node12);
\draw (node11) [line width=0.5pt] -- node [left] {} (node12);
\draw (node9) [line width=0.5pt] -- node [below] {} (node12);
\draw (node12) [line width=0.5pt] -- node [below] {} (node13);
\draw (node13) [line width=0.5pt] -- node [left] {} (node14);
\draw (node14) [line width=0.5pt] -- node [left] {} (node15);
\draw (node15) [line width=0.5pt] -- node [left] {} (node16);
\draw (node13) [line width=0.5pt] -- node [left] {} (node16);

\node[place, circle] (node17) at (-0.5,4) [label=below:\scriptsize$u_{1}$] {}; 
\node[place, circle] (node18) at (3.5,2) [label=below:\scriptsize$u_{2}$] {}; 
%\node[place, circle] (node20) at (6,9) [label=below:\scriptsize$u_{3}$] {}; 
%\node[place, circle] (node19) at (13.5,7) [label=below:\scriptsize$u_{4}$] {}; 
\draw (node17) [-latex, line width=0.5pt] -- node [right] {} (node1);
\draw (node18) [-latex, line width=0.5pt] -- node [right] {} (node3);
%\draw (node7) [-latex, line width=0.5pt] -- node [right] {} (node20);
%\draw (node14) [-latex, line width=0.5pt] -- node [right] {} (node19);
\end{tikzpicture}
}
\subfigure[Step 2 : $\bar{\mathcal{G}}_{2}=\mathcal{G}_{2}\cup\mathcal{G}_{23}$]{
\begin{tikzpicture}[scale=0.5]

%\node[] at (2.9,4.7) {\scriptsize$\mathcal{G}_1$};
\node[] at (7,6) {\scriptsize$\mathcal{G}_2$};
\node[] at (10.4,3) {\scriptsize$\mathcal{G}_3$};
\node[] at (14,4.5) {\scriptsize$\mathcal{G}_4$};

\node[] at (4.7,4.2) {\scriptsize$\mathcal{G}_{12}$};
\node[] at (9,4) {\scriptsize$\mathcal{G}_{23}$};
\node[] at (12.2,4) {\scriptsize$\mathcal{G}_{34}$};

%\node[place, blue] (node1) at (1,5) [label=below:\scriptsize$1$] {};
\node[place, circle] (node2) at (3,4) [label=below:\scriptsize$2$] {};
\node[place, circle] (node3) at (5,3) [label=below:\scriptsize$3$] {};
\node[place, blue] (node4) at (7,4) [label=below:\scriptsize$4$] {};
\node[place, blue] (node5) at (5,5) [label=right:\scriptsize$5$] {};
\node[place, blue] (node6) at (5,7) [label=above:\scriptsize$6$] {};
\node[place, blue] (node7) at (7,8) [label=above:\scriptsize$7$] {};
\node[place, blue] (node8) at (9,7) [label=above:\scriptsize$8$] {};
\node[place, blue] (node9) at (9,5) [label=left:\scriptsize$9$] {};
\node[place, blue] (node10) at (9,3) [label=below:\scriptsize$10$] {};
\node[place, black] (node11) at (11,2) [label=below:\scriptsize$11$] {};
\node[place, blue] (node12) at (11,4) [label=above:\scriptsize$12$] {};
\node[place, black] (node13) at (13,5) [label=above:\scriptsize$13$] {};
\node[place, black] (node14) at (15,6) [label=above:\scriptsize$14$] {};
\node[place, black] (node15) at (15,4) [label=below:\scriptsize$15$] {};
\node[place, black] (node16) at (13,3) [label=below:\scriptsize$16$] {};

%\draw (node1) [line width=0.5pt] -- node [left] {} (node2);
%\draw (node2) [line width=0.5pt] -- node [left] {} (node3);
\draw (node2) [line width=0.5pt] -- node [below] {} (node6);
\draw (node3) [line width=0.5pt] -- node [below] {} (node4);
\draw (node4) [line width=0.5pt] -- node [left] {} (node5);
\draw (node5) [line width=0.5pt] -- node [left] {} (node6);
\draw (node6) [line width=0.5pt] -- node [left] {} (node7);
\draw (node7) [line width=0.5pt] -- node [left] {} (node8);
\draw (node8) [line width=0.5pt] -- node [left] {} (node9);
\draw (node4) [line width=0.5pt] -- node [left] {} (node9);
\draw (node4) [line width=0.5pt] -- node [below] {} (node10);
\draw (node10) [line width=0.5pt] -- node [left] {} (node11);
\draw (node10) [line width=0.5pt] -- node [left] {} (node12);
\draw (node11) [line width=0.5pt] -- node [left] {} (node12);
\draw (node9) [line width=0.5pt] -- node [below] {} (node12);
\draw (node12) [line width=0.5pt] -- node [below] {} (node13);
\draw (node13) [line width=0.5pt] -- node [left] {} (node14);
\draw (node14) [line width=0.5pt] -- node [left] {} (node15);
\draw (node15) [line width=0.5pt] -- node [left] {} (node16);
\draw (node13) [line width=0.5pt] -- node [left] {} (node16);

%\node[place, circle] (node17) at (-0.5,4) [label=below:\scriptsize$u_{1}$] {}; 
%\node[place, circle] (node18) at (3.5,2) [label=below:\scriptsize$u_{2}$] {}; 
\node[place, circle] (node20) at (5.5,9) [label=below:\scriptsize$u_{3}$] {}; 
%\node[place, circle] (node19) at (13.5,7) [label=below:\scriptsize$u_{4}$] {}; 
%\draw (node1) [-latex, line width=0.5pt] -- node [right] {} (node17);
%\draw (node3) [-latex, line width=0.5pt] -- node [right] {} (node18);
\draw (node20) [-latex, line width=0.5pt] -- node [right] {} (node7);
%\draw (node14) [-latex, line width=0.5pt] -- node [right] {} (node19);
\end{tikzpicture}
}

\subfigure[Step 3 : $\bar{\mathcal{G}}_{3}=\mathcal{G}_{3}\cup\mathcal{G}_{34}$]{
\begin{tikzpicture}[scale=0.5]

%\node[] at (2.9,4.7) {\scriptsize$\mathcal{G}_1$};
%\node[] at (7,6) {\scriptsize$\mathcal{G}_2$};
\node[] at (10.4,3) {\scriptsize$\mathcal{G}_3$};
\node[] at (14,4.5) {\scriptsize$\mathcal{G}_4$};

%\node[] at (4.7,4.2) {\scriptsize$\mathcal{G}_{12}$};
\node[] at (9,4) {\scriptsize$\mathcal{G}_{23}$};
\node[] at (12.2,4) {\scriptsize$\mathcal{G}_{34}$};

%\node[place, blue] (node1) at (1,5) [label=below:\scriptsize$1$] {};
%\node[place, circle] (node2) at (3,4) [label=below:\scriptsize$2$] {};
%\node[place, circle] (node3) at (5,3) [label=below:\scriptsize$3$] {};
\node[place, circle] (node4) at (7,4) [label=above:\scriptsize$4$] {};
%\node[place, blue] (node5) at (5,5) [label=right:\scriptsize$5$] {};
%\node[place, blue] (node6) at (5,7) [label=above:\scriptsize$6$] {};
%\node[place, blue] (node7) at (7,8) [label=above:\scriptsize$7$] {};
%\node[place, blue] (node8) at (9,7) [label=above:\scriptsize$8$] {};
\node[place, circle] (node9) at (9,5) [label=above:\scriptsize$9$] {};
\node[place, blue] (node10) at (9,3) [label=below:\scriptsize$10$] {};
\node[place, blue] (node11) at (11,2) [label=below:\scriptsize$11$] {};
\node[place, blue] (node12) at (11,4) [label=above:\scriptsize$12$] {};
\node[place, blue] (node13) at (13,5) [label=above:\scriptsize$13$] {};
\node[place, black] (node14) at (15,6) [label=above:\scriptsize$14$] {};
\node[place, black] (node15) at (15,4) [label=below:\scriptsize$15$] {};
\node[place, black] (node16) at (13,3) [label=below:\scriptsize$16$] {};

%\draw (node1) [line width=0.5pt] -- node [left] {} (node2);
%\draw (node2) [line width=0.5pt] -- node [left] {} (node3);
%\draw (node2) [line width=0.5pt] -- node [below] {} (node6);
%\draw (node3) [line width=0.5pt] -- node [below] {} (node4);
%\draw (node4) [line width=0.5pt] -- node [left] {} (node5);
%\draw (node5) [line width=0.5pt] -- node [left] {} (node6);
%\draw (node6) [line width=0.5pt] -- node [left] {} (node7);
%\draw (node7) [line width=0.5pt] -- node [left] {} (node8);
%\draw (node8) [line width=0.5pt] -- node [left] {} (node9);
%\draw (node4) [line width=0.5pt] -- node [left] {} (node9);
\draw (node4) [line width=0.5pt] -- node [below] {} (node10);
\draw (node10) [line width=0.5pt] -- node [left] {} (node11);
\draw (node10) [line width=0.5pt] -- node [left] {} (node12);
\draw (node11) [line width=0.5pt] -- node [left] {} (node12);
\draw (node9) [line width=0.5pt] -- node [below] {} (node12);
\draw (node12) [line width=0.5pt] -- node [below] {} (node13);
\draw (node13) [line width=0.5pt] -- node [left] {} (node14);
\draw (node14) [line width=0.5pt] -- node [left] {} (node15);
\draw (node15) [line width=0.5pt] -- node [left] {} (node16);
\draw (node13) [line width=0.5pt] -- node [left] {} (node16);

%\node[place, circle] (node17) at (-0.5,4) [label=below:\scriptsize$u_{1}$] {}; 
%\node[place, circle] (node18) at (3.5,2) [label=below:\scriptsize$u_{2}$] {}; 
%\node[place, circle] (node20) at (6,9) [label=below:\scriptsize$u_{3}$] {}; 
%\node[place, circle] (node19) at (13.5,7) [label=below:\scriptsize$u_{4}$] {}; 
%\draw (node1) [-latex, line width=0.5pt] -- node [right] {} (node17);
%\draw (node3) [-latex, line width=0.5pt] -- node [right] {} (node18);
%\draw (node7) [-latex, line width=0.5pt] -- node [right] {} (node20);
%\draw (node14) [-latex, line width=0.5pt] -- node [right] {} (node19);
\end{tikzpicture}

}
\subfigure[Step 4 : $\bar{\mathcal{G}}_{4}=\mathcal{G}_{4}$]{
\begin{tikzpicture}[scale=0.5]

%\node[] at (2.9,4.7) {\scriptsize$\mathcal{G}_1$};
%\node[] at (7,6) {\scriptsize$\mathcal{G}_2$};
%\node[] at (10.4,3) {\scriptsize$\mathcal{G}_3$};
\node[] at (14,4.5) {\scriptsize$\mathcal{G}_4$};

%\node[] at (4.7,4.2) {\scriptsize$\mathcal{G}_{12}$};
%\node[] at (9,4) {\scriptsize$\mathcal{G}_{23}$};
\node[] at (12.2,4) {\scriptsize$\mathcal{G}_{34}$};

%\node[place, blue] (node1) at (1,5) [label=below:\scriptsize$1$] {};
%\node[place, circle] (node2) at (3,4) [label=below:\scriptsize$2$] {};
%\node[place, circle] (node3) at (5,3) [label=below:\scriptsize$3$] {};
%\node[place, circle] (node4) at (7,4) [label=above:\scriptsize$4$] {};
%\node[place, blue] (node5) at (5,5) [label=right:\scriptsize$5$] {};
%\node[place, blue] (node6) at (5,7) [label=above:\scriptsize$6$] {};
%\node[place, blue] (node7) at (7,8) [label=above:\scriptsize$7$] {};
%\node[place, blue] (node8) at (9,7) [label=above:\scriptsize$8$] {};
%\node[place, circle] (node9) at (9,5) [label=above:\scriptsize$9$] {};
%\node[place, blue] (node10) at (9,3) [label=below:\scriptsize$10$] {};
%\node[place, blue] (node11) at (11,2) [label=below:\scriptsize$11$] {};
\node[place, circle] (node12) at (11,4) [label=above:\scriptsize$12$] {};
\node[place, blue] (node13) at (13,5) [label=above:\scriptsize$13$] {};
\node[place, blue] (node14) at (15,6) [label=above:\scriptsize$14$] {};
\node[place, blue] (node15) at (15,4) [label=below:\scriptsize$15$] {};
\node[place, blue] (node16) at (13,3) [label=below:\scriptsize$16$] {};

%\draw (node1) [line width=0.5pt] -- node [left] {} (node2);
%\draw (node2) [line width=0.5pt] -- node [left] {} (node3);
%\draw (node2) [line width=0.5pt] -- node [below] {} (node6);
%\draw (node3) [line width=0.5pt] -- node [below] {} (node4);
%\draw (node4) [line width=0.5pt] -- node [left] {} (node5);
%\draw (node5) [line width=0.5pt] -- node [left] {} (node6);
%\draw (node6) [line width=0.5pt] -- node [left] {} (node7);
%\draw (node7) [line width=0.5pt] -- node [left] {} (node8);
%\draw (node8) [line width=0.5pt] -- node [left] {} (node9);
%\draw (node4) [line width=0.5pt] -- node [left] {} (node9);
%\draw (node4) [line width=0.5pt] -- node [below] {} (node10);
%\draw (node10) [line width=0.5pt] -- node [left] {} (node11);
%\draw (node10) [line width=0.5pt] -- node [left] {} (node12);
%\draw (node11) [line width=0.5pt] -- node [left] {} (node12);
%\draw (node9) [line width=0.5pt] -- node [below] {} (node12);
\draw (node12) [line width=0.5pt] -- node [below] {} (node13);
\draw (node13) [line width=0.5pt] -- node [left] {} (node14);
\draw (node14) [line width=0.5pt] -- node [left] {} (node15);
\draw (node15) [line width=0.5pt] -- node [left] {} (node16);
\draw (node13) [line width=0.5pt] -- node [left] {} (node16);

%\node[place, circle] (node17) at (-0.5,4) [label=below:\scriptsize$u_{1}$] {}; 
%\node[place, circle] (node18) at (3.5,2) [label=below:\scriptsize$u_{2}$] {}; 
%\node[place, circle] (node20) at (6,9) [label=below:\scriptsize$u_{3}$] {}; 
\node[place, circle] (node19) at (13.5,7) [label=below:\scriptsize$u_{4}$] {}; 
%\draw (node1) [-latex, line width=0.5pt] -- node [right] {} (node17);
%\draw (node3) [-latex, line width=0.5pt] -- node [right] {} (node18);
%\draw (node7) [-latex, line width=0.5pt] -- node [right] {} (node20);
\draw (node19) [-latex, line width=0.5pt] -- node [right] {} (node14);
\end{tikzpicture}
}

\subfigure[Final : $\mathcal{G}^{out}=\bar{\mathcal{G}}_{1}\cup\bar{\mathcal{G}}_{2}\cup\bar{\mathcal{G}}_{3}\cup\bar{\mathcal{G}}_{4}$]{
\begin{tikzpicture}[scale=0.5]

\node[] at (2.9,4.7) {\scriptsize$\mathcal{G}_1$};
\node[] at (7,6) {\scriptsize$\mathcal{G}_2$};
\node[] at (10.4,3) {\scriptsize$\mathcal{G}_3$};
\node[] at (14,4.5) {\scriptsize$\mathcal{G}_4$};

\node[] at (4.7,4.2) {\scriptsize$\mathcal{G}_{12}$};
\node[] at (9,4) {\scriptsize$\mathcal{G}_{23}$};
\node[] at (12.2,4) {\scriptsize$\mathcal{G}_{34}$};

\node[place, blue] (node1) at (1,5) [label=below:\scriptsize$1$] {};
\node[place, blue] (node2) at (3,4) [label=below:\scriptsize$2$] {};
\node[place, blue] (node3) at (5,3) [label=below:\scriptsize$3$] {};
\node[place, blue] (node4) at (7,4) [label=below:\scriptsize$4$] {};
\node[place, blue] (node5) at (5,5) [label=right:\scriptsize$5$] {};
\node[place, blue] (node6) at (5,7) [label=above:\scriptsize$6$] {};
\node[place, blue] (node7) at (7,8) [label=above:\scriptsize$7$] {};
\node[place, blue] (node8) at (9,7) [label=above:\scriptsize$8$] {};
\node[place, blue] (node9) at (9,5) [label=left:\scriptsize$9$] {};
\node[place, blue] (node10) at (9,3) [label=below:\scriptsize$10$] {};
\node[place, blue] (node11) at (11,2) [label=below:\scriptsize$11$] {};
\node[place, blue] (node12) at (11,4) [label=above:\scriptsize$12$] {};
\node[place, blue] (node13) at (13,5) [label=above:\scriptsize$13$] {};
\node[place, blue] (node14) at (15,6) [label=above:\scriptsize$14$] {};
\node[place, blue] (node15) at (15,4) [label=below:\scriptsize$15$] {};
\node[place, blue] (node16) at (13,3) [label=below:\scriptsize$16$] {};

\draw (node1) [line width=0.5pt] -- node [left] {} (node2);
\draw (node2) [line width=0.5pt] -- node [left] {} (node3);
\draw (node2) [line width=0.5pt] -- node [below] {} (node6);
\draw (node3) [,line width=0.5pt] -- node [below] {} (node4);
\draw (node4) [line width=0.5pt] -- node [left] {} (node5);
\draw (node5) [line width=0.5pt] -- node [left] {} (node6);
\draw (node6) [line width=0.5pt] -- node [left] {} (node7);
\draw (node7) [line width=0.5pt] -- node [left] {} (node8);
\draw (node8) [line width=0.5pt] -- node [left] {} (node9);
\draw (node4) [line width=0.5pt] -- node [left] {} (node9);
\draw (node4) [line width=0.5pt] -- node [below] {} (node10);
\draw (node10) [line width=0.5pt] -- node [left] {} (node11);
\draw (node10) [line width=0.5pt] -- node [left] {} (node12);
\draw (node11) [line width=0.5pt] -- node [left] {} (node12);
\draw (node9) [line width=0.5pt] -- node [below] {} (node12);
\draw (node12) [line width=0.5pt] -- node [below] {} (node13);
\draw (node13) [line width=0.5pt] -- node [left] {} (node14);
\draw (node14) [line width=0.5pt] -- node [left] {} (node15);
\draw (node15) [line width=0.5pt] -- node [left] {} (node16);
\draw (node13) [line width=0.5pt] -- node [left] {} (node16);
 
\node[place, circle] (node17) at (-0.5,4) [label=below:\scriptsize$u_{1}$] {}; 
\node[place, circle] (node18) at (3.5,2) [label=below:\scriptsize$u_{2}$] {}; 
\node[place, circle] (node20) at (5.5,9) [label=below:\scriptsize$u_{3}$] {}; 
\node[place, circle] (node19) at (13.5,7) [label=below:\scriptsize$u_{4}$] {}; 
\draw (node17) [-latex, line width=0.5pt] -- node [right] {} (node1); 
\draw (node18) [-latex, line width=0.5pt] -- node [right] {} (node3);
\draw (node20) [-latex, line width=0.5pt] -- node [right] {} (node7);
\draw (node19) [-latex, line width=0.5pt] -- node [right] {} (node14);
\end{tikzpicture}
}
\caption{Topological example of \textit{Algorithm~\ref{graph_merging}} for the sym-pactus $\mathcal{G}^{S}$ shown in Fig.~\ref{network_ex_final}. The blue marked nodes are the SSC state nodes. }
\label{network_ex_step}
\end{figure}

\subsection{Numerical verification} 
For a numerical verification of the strong sign controllability of a network, 
it is necessary to verify that the rank of controllability Gramian matrix of a given network has full rank for all choices of weights under the sign-fixed condition.
In this subsection, we verify that the output graph shown in Fig.~\ref{network_ex_step}.(e) has full rank for all choices of weights.
The controllability Gramian matrix corresponding to $L$ and $B$ is: %reduced row echelon form (RREF)
\begin{align}
\mathcal{C}_{L} = [B, LB, L^2B, ... ,L^{14}B, L^{15}B]
\end{align}
where $L\in\Bbb{R}^{16\times 16}$ is the Laplacian matrix and $B\in\Bbb{R}^{16\times 4}$ is the external input matrix of the given network shown in Fig.~\ref{network_ex_step}.(e).
To check the rank of $\mathcal{C}_{L}$, it needs to be processed using the well-known Gauss elimination method with one additional condition during row operation.
There are three types of elementary row operations under the \textit{pivot}\footnote{If a matrix is row-echelon form, then the first nonzero entry of each row is called a pivot.} 
condition:
\begin{condition} \label{condition_pivot}
If \textit{pivot} element consists of two-or-more-terms, set that \textit{pivot} to zero.
\begin{itemize}
\item Swapping two rows
\item Multiplying a row by a non-zero \textit{pivot} element
\item Adding a multiple of one row to another row
\end{itemize}
\end{condition}

In \textit{Condition~\ref{condition_pivot}}, we assume that all equations of elements of $\mathcal{C}_{L}$ are simplified.
Let us suppose that a pivot element consists of only one-term, e.g., $a_{12}$ or $a_{12}a_{23}$, obviously, it cannot be zero. 
However, if a pivot element consists of two-or-more-terms, it has a possibility to be zero by applying a transposition. 
Therefore, the above \textit{Condition~\ref{condition_pivot}} means that if a pivot has a possibility to be zero, we assume the value of that pivot is zero.
This condition prevents \textit{generic property} that a network being uncontrollable for specific choices of weights.
As a result, the row rank of $\mathcal{C}_{L}$ for the given network shown in Fig.~\ref{network_ex_step}(e) was 16. 
As a simple example of \textit{Condition~\ref{condition_pivot}}, let us consider a sym-path state graph with one external input node, 
then $L$, $B$ and the corresponding $\mathcal{C}_{L}$ are following.
\begin{align} \label{ex_Laplacian}
L = \left[\begin{array}{ccccc}
               -\Sigma_1  & a_{12} & 0    \\
              a_{12} &  -\Sigma_2  & a_{23}  \\
               0 & a_{23} &  -\Sigma_3     \end{array} \right],
B = \left[\begin{array}{ccccc}
               0  \\
              1  \\
              0 \end{array} \right]
\end{align}
\begin{align} \label{ex_gramian}
\mathcal{C}_{L} = \left[\begin{array}{ccccc}
              0 & a_{12} &-a_{12}(\Sigma_1  + \Sigma_2)       \\
              1 & -\Sigma_2 & a_{12}^2 + (\Sigma_2)^2 + a_{23}^2\\
              0 & a_{23} &-a_{23}(\Sigma_2 + \Sigma_3)    \end{array} \right]
\end{align}
where $\Sigma_1=a_{12}$, $\Sigma_2=a_{12}+a_{23}$, and $\Sigma_3=a_{23}$. 
After the Gauss elimination process with \textit{Condition~\ref{condition_pivot}}, the reduced row echelon form (RREF) of (\ref{ex_gramian}) is given as $\bar{\mathcal{C}}_{L}$:
\begin{align} \label{ex_gramian_RREF}
\bar{\mathcal{C}}_{L} = \left[\begin{array}{ccccc}
               1 & \Sigma_2 & a_{12}^2 + (\Sigma_2)^2 + a_{23}^2\\
              0 & 1 & -(\Sigma_1 + \Sigma_2)   \\
              0 & 0 & \boldsymbol{\Sigma_1 - \Sigma_3}   \end{array} \right]
\end{align}
In matrix \eqref{ex_gramian_RREF}, the pivots of the first row $\bar{\mathcal{C}}_{L}(1,1)$ and the second row $\bar{\mathcal{C}}_{L}(2,2)$ are non-zero because those are one-term.
However, the pivot of the third row $\bar{\mathcal{C}}_{L}(3,3)$ has two-or-more-terms as $\Sigma_1-\Sigma_3=a_{12}-a_{23}$. 
According to \textit{Condition~\ref{condition_pivot}}, the pivot $\bar{\mathcal{C}}_{L}(3,3)$ is set to zero
since it can be zero when $a_{12}=a_{23}$.
It means that a specific choice of weight makes the given network uncontrollable.
Then, the rank of the output matrix of the Gauss elimination process is 2, which is not full rank. 
Therefore, we can conclude that the given network $T=[L,B]$ is structurally controllable, but not strongly sign controllable (or, not strongly structurally controllable).

\section{Conclusion} \label{sec_conc}
This paper presented the conditions that determine the strong sign controllability for undirected signed networks of diffusively-coupled agents.
First, we established the concepts of dedicated and sharing nodes for the strong sign controllability, and the condition of strong sign controllability of a state node unit is presented.
Then, we interpret the existing results on the controllability condition of basic components (path, tree, cycle) from the strong sign controllability point of view based on the dedicated and sharing nodes.
To extend the results into a larger graph, e.g., sym-pactus, we propose an algorithm with a sufficient number of external input nodes.
Furthermore, we interpreted the property of external input nodes from the dedicated node perspective, 
and applied the notion of component input node to the problem of reducing the number of external input nodes.
In particular, we established an algorithm with a polynomial complexity for the strong sign controllability of a sym-pactus to find the minimum number of external input nodes.
As a dual problem, our analysis of the strong sign controllability can be extended to strong sign observability. 
For example, the strong sign observability can be applied to the problems of evaluating the minimal number of measurements to estimate states in a diffusive network.
Moreover, from the perspective of a state unit, e.g., the SSC state nodes presented in this paper, 
the concepts of strong sign controllability and observability can be extended to a Kalman decomposition \cite{hovelaque1997kalman} for diffusively-coupled signed networks.
These extensions will be studied in our future work.

%\section*{Acknowledgment}
%\textcolor{red}{
%The work of this paper has been supported by the National Research Foundation (NRF) of Korea under the grant NRF2017R1A2B3007034.}
%The authors would like to appreciate Byung-Hun Lee for his comments and the discussions on the theoretical results.

%\section*{Acknowledgment}
%The works of this paper have been supported by the National Research Foundation (NRF) of Korea under the grant NRF-2017R1A2B3007034.

\bibliographystyle{IEEEtran}
\bibliography{references}

\begin{thebibliography}{10}
\providecommand{\url}[1]{#1}
\csname url@samestyle\endcsname
\providecommand{\newblock}{\relax}
\providecommand{\bibinfo}[2]{#2}
\providecommand{\BIBentrySTDinterwordspacing}{\spaceskip=0pt\relax}
\providecommand{\BIBentryALTinterwordstretchfactor}{4}
\providecommand{\BIBentryALTinterwordspacing}{\spaceskip=\fontdimen2\font plus
\BIBentryALTinterwordstretchfactor\fontdimen3\font minus
  \fontdimen4\font\relax}
\providecommand{\BIBforeignlanguage}[2]{{%
\expandafter\ifx\csname l@#1\endcsname\relax
\typeout{** WARNING: IEEEtran.bst: No hyphenation pattern has been}%
\typeout{** loaded for the language `#1'. Using the pattern for}%
\typeout{** the default language instead.}%
\else
\language=\csname l@#1\endcsname
\fi
#2}}
\providecommand{\BIBdecl}{\relax}
\BIBdecl

\bibitem{ahn2019topological}
H.-S. Ahn, K.~L. Moore, S.-H. Kwon, Q.~Van~Tran, B.-Y. Kim, and K.-K. Oh,
  ``Topological controllability of undirected networks of diffusively-coupled
  agents,'' in \emph{2019 58th Annual Conference of the Society of Instrument
  and Control Engineers of Japan (SICE)}.\hskip 1em plus 0.5em minus
  0.4em\relax IEEE, 2019, pp. 673--678.

\bibitem{ahn2019topological_arXiv}
------, ``Topological controllability of undirected networks of
  diffusively-coupled agents,'' \emph{arXiv preprint arXiv:1903.11246}, 2019.

\bibitem{lin1974structural}
C.-T. Lin, ``Structural controllability,'' \emph{IEEE Transactions on Automatic
  Control}, vol.~19, no.~3, pp. 201--208, 1974.

\bibitem{mayeda1979strong}
H.~Mayeda and T.~Yamada, ``Strong structural controllability,'' \emph{SIAM
  Journal on Control and Optimization}, vol.~17, no.~1, pp. 123--138, 1979.

\bibitem{johnson1993sign}
C.~R. Johnson, V.~Mehrmann, and D.~D. Olesky, ``Sign controllability of a
  nonnegative matrix and a positive vector,'' \emph{SIAM journal on matrix
  analysis and applications}, vol.~14, no.~2, pp. 398--407, 1993.

\bibitem{mousavi2019strong}
S.~S. Mousavi, M.~Haeri, and M.~Mesbahi, ``Strong structural controllability of
  signed networks,'' in \emph{2019 IEEE 58th Conference on Decision and Control
  (CDC)}.\hskip 1em plus 0.5em minus 0.4em\relax IEEE, 2019, pp. 4557--4562.

\bibitem{guan2021structural}
Y.~Guan, A.~Li, and L.~Wang, ``Structural controllability of directed signed
  networks,'' \emph{IEEE Transactions on Control of Network Systems}, vol.~8,
  no.~3, pp. 1189--1200, 2021.

\bibitem{she2018controllability}
B.~She, S.~Mehta, C.~Ton, and Z.~Kan, ``Controllability ensured leader group
  selection on signed multiagent networks,'' \emph{IEEE Transactions on
  Cybernetics}, vol.~50, no.~1, pp. 222--232, 2018.

\bibitem{da2018topology}
B.~R. da~Cunha and S.~Gon{\c{c}}alves, ``Topology, robustness, and structural
  controllability of the brazilian federal police criminal intelligence
  network,'' \emph{Applied network science}, vol.~3, no.~1, pp. 1--20, 2018.

\bibitem{menara2018structural}
T.~Menara, D.~S. Bassett, and F.~Pasqualetti, ``Structural controllability of
  symmetric networks,'' \emph{IEEE Transactions on Automatic Control}, vol.~64,
  no.~9, pp. 3740--3747, 2018.

\bibitem{goldstein1979controllability}
L.~Goldstein, ``Controllability/observability analysis of digital circuits,''
  \emph{IEEE Transactions on Circuits and Systems}, vol.~26, no.~9, pp.
  685--693, 1979.

\bibitem{feng2005structural}
X.-Y. Feng and K.-S. Lu, ``Structural controllability and reducibility of rlc
  networks with bipolar transistor,'' in \emph{2005 International Conference on
  Machine Learning and Cybernetics}, vol.~2.\hskip 1em plus 0.5em minus
  0.4em\relax IEEE, 2005, pp. 1015--1020.

\bibitem{jia2020unifying}
J.~Jia, H.~J. van Waarde, H.~L. Trentelman, and M.~K. Camlibel, ``A unifying
  framework for strong structural controllability,'' \emph{IEEE Transactions on
  Automatic Control}, vol.~66, no.~1, pp. 391--398, 2020.

\bibitem{monshizadeh2014zero}
N.~Monshizadeh, S.~Zhang, and M.~K. Camlibel, ``Zero forcing sets and
  controllability of dynamical systems defined on graphs,'' \emph{IEEE
  Transactions on Automatic Control}, vol.~59, no.~9, pp. 2562--2567, 2014.

\bibitem{chapman2013strong}
A.~Chapman and M.~Mesbahi, ``On strong structural controllability of networked
  systems: A constrained matching approach,'' in \emph{2013 American Control
  Conference}.\hskip 1em plus 0.5em minus 0.4em\relax IEEE, 2013, pp.
  6126--6131.

\bibitem{trefois2015zero}
M.~Trefois and J.-C. Delvenne, ``Zero forcing number, constrained matchings and
  strong structural controllability,'' \emph{Linear Algebra and its
  Applications}, vol. 484, pp. 199--218, 2015.

\bibitem{mousavi2017structural}
S.~S. Mousavi, M.~Haeri, and M.~Mesbahi, ``On the structural and strong
  structural controllability of undirected networks,'' \emph{IEEE Transactions
  on Automatic Control}, vol.~63, no.~7, pp. 2234--2241, 2017.

\bibitem{tsatsomeros1998sign}
M.~J. Tsatsomeros, ``Sign controllability: Sign patterns that require complete
  controllability,'' \emph{SIAM journal on matrix analysis and applications},
  vol.~19, no.~2, pp. 355--364, 1998.

\bibitem{hartung2014sign}
C.~Hartung and F.~Svaricek, ``Sign stabilizability,'' in \emph{22nd
  Mediterranean Conference on Control and Automation}.\hskip 1em plus 0.5em
  minus 0.4em\relax IEEE, 2014, pp. 145--150.

\bibitem{she2020energy}
B.~She, S.~Mehta, C.~Ton, and Z.~Kan, ``Energy-related controllability of
  signed complex networks with laplacian dynamics,'' \emph{IEEE Transactions on
  Automatic Control}, vol.~66, no.~7, pp. 3325--3330, 2020.

\bibitem{rugh1996linear}
W.~J. Rugh, \emph{Linear system theory}.\hskip 1em plus 0.5em minus 0.4em\relax
  Prentice-Hall, Inc., 1996.

\bibitem{hovelaque1997kalman}
V.~Hovelaque, C.~Commault, and J.~Dion, ``Kalman decomposition of linear
  structured systems,'' \emph{IFAC Proceedings Volumes}, vol.~30, no.~27, pp.
  81--86, 1997.

\end{thebibliography}
\end{document}